\documentclass[preprint, prd, onecolumn]{revtex4}

\usepackage{amsmath}
\usepackage{graphicx}
\usepackage{dcolumn}
\usepackage{bm}
\usepackage{wrapfig}
\usepackage{epsfig}
\usepackage{breqn}
\usepackage{amssymb}
\usepackage{subfigure}
\usepackage[usenames,dvipsnames]{pstricks}
\usepackage{epsfig}
\usepackage{pst-grad} 
\usepackage{pst-plot} 
 \usepackage[usenames,dvipsnames]{pstricks}
 \usepackage{epsfig}
 \usepackage{pst-grad} 
 \usepackage{pst-plot} 
 \usepackage{float}

\newtheorem{prop}{Proposition}
\newenvironment{proof}[1][Proof]{\noindent\textbf{#1.} }{\ \rule{0.5em}{0.5em}}

\begin{document}

\preprint{APS/123-QED}

\title{A Dynamical Systems Approach to a \\ Bianchi Type I Viscous Magnetohydrodynamic Model}
\author{Ikjyot Singh Kohli}
\email{isk@yorku.ca}
\affiliation{York University - Department of Physics and Astronomy}
\author{Michael C. Haslam}
\email{mchaslam@mathstat.yorku.ca}
\affiliation{
York University - Department of Mathematics and Statistics
}

\date{April 1, 2013}

\begin{abstract}
We use the expansion-normalized variables approach to study the dynamics of a non-tilted Bianchi Type I cosmological model with both a homogeneous magnetic field and a viscous fluid. In our model the perfect magnetohydrodynamic approximation is made, and both bulk and shear viscous effects are retained. The dynamical system is studied in detail through a fixed-point analysis which determines the local sink and source behavior of the system. We show that the fixed points may be associated with Kasner-type solutions, a flat universe FLRW solution, and interestingly, a new solution to the Einstein Field equations involving non-zero magnetic fields, and non-zero viscous coefficients. 
It is further shown that for certain values of the bulk and shear viscosity and equation of state parameters, the model isotropizes at late times.
\end{abstract}
\maketitle 


\section{Introduction}

The current standard model of cosmology based on the Friedmann-LeMaitre-Robertson-Walker (FLRW) metric assumes that the present-day universe is spatially homogeneous and isotropic, and indeed this assumption strongly concurs with empirical observation. 
As a result of the symmetry of this spacetime, related models must be treated within the framework of perfect fluids, in which case the shear and rotational terms in the energy-momentum tensor vanish (page 52, \cite{ellis}).

If one wishes to formulate a cosmological model of the early universe, however, at a minimum it is necessary to include viscous (shear) terms in the energy-momentum tensor. 
As discussed by Gr{\o}n and Hervik (Chapter 13, \cite{hervik}), viscous models have become of general interest in early-universe cosmologies largely in two contexts.  Firstly, in models where bulk viscous terms dominate over shear terms, the universe expands to a de Sitter-like state, which is a spatially flat universe neglecting ordinary matter, and including only a cosmological constant. Such models isotropize indirectly through the massive expansion. Secondly, in the absence of any significant heat flux, shear viscosity is found to play an important role in models of the universe at its early stages. In particular, neutrino viscosity is considered to be one of the most important factors in the isotropization of our universe.  

Magnetic fields have also been thought to play a major role in the early universe. Grasso and Rubinstein \cite{grassorub} reviewed in great detail the origin and possible effects of magnetic fields in the early universe. In recent work, Ando and Kusenko \cite{andokusenko}, examined intergalactic magnetic fields and discussed how these magnetic fields originated from primordial seed fields created shortly after the big bang, which relates to our understanding of the origin of cosmic magnetic fields in the early universe. In addition, Gregori et al. \cite{gregorietal} also studied the origin of galactic magnetic fields through the amplification of primordial seed fields. Schlickeiser \cite{Schlickeiser} described a new process by which the primordial magnetic fields arose in the universe before the emergence of the first stars.

After inflation the early universe was a good conductor: even though the number density of free electrons dropped dramatically during recombination, its residual value was enough to maintain high conductivity in baryonic matter. As a result, cosmic magnetic fields have remained frozen into the expanding baryonic fluid during most of their evolution. In this situation, one can analyze the magnetic effects on the dynamics of the early universe through ideal magnetohydrodynamics (hereafter referred to as MHD), in which case the magnetic field source is considered to be a perfect conductor and related terms in the energy momentum tensor are simply those corresponding to a classical magnetic field (Page 115, \cite{elliscosmo}).

Hughston and Jacobs \cite{hughstonjacobs} showed that in the case of a pure magnetic field, only Bianchi Types I, II, VI($h = -1$) (which is the same as Type III), and VII ($h=0$) admit field components, whereas Types IV, V, VI ($h=-1$), VII ($h\neq0$), VIII, and IX admit no field components. These results led to a number of papers of Bianchi models with a perfect-fluid magnetic field source; we discuss these works briefly below. Using a dynamical systems approach LeBlanc \cite{leblanc1} studied Bianchi Type II magnetic cosmologies in which he provided an analysis on the future and past asymptotic states of the resulting dynamical system. In a separate work, LeBlanc \cite{leblanc2} also studied the asymptotic states of magnetic perfect-fluid Bianchi Type I cosmologies. Using phase plane analysis techniques, Collins \cite{collins} studied the behavior of a class of perfect-fluid anisotropic cosmological models, and established a correspondence between magnetic models of Bianchi Type I and perfect fluid models of Bianchi Type II. In addition, LeBlanc, Kerr, and Wainwright \cite{leblanc3} studied the asymptotic states of magnetic Bianchi Type VI cosmologies and showed that there is a finite probability that an arbitrarily selected model will be close to isotropy during some time interval in its evolution. We also note that Barrow, Maartens, and Tsagas \cite{barrowmaartenstsagas} did significant work in the reformulation of a $1+3$ covariant description of the magnetohydrodynamic equations that has provided further understanding and clarity on the role of large-scale electromagnetic fields in the perturbed Friedmann-LeMaitre-Robertson-Walker models.

Viscous MHD Bianchi models treated using a metric approach have appeared in the literature on a number of occasions. van Leeuwen and Salvati \cite{vanLeeuwen1} studied the dynamics of general Bianchi class A models containing a magneto-viscous fluid and a large-scale magnetic field. Banerjee and Sanyal \cite{banerjeesanyal} presented some exact solutions of Bianchi Types I and III cosmological models consisting of a viscous fluid and axial magnetic field. Benton and Tupper \cite{bentontupper} studied Bianchi Type I models with a ``powers-of-t'' metric under the influence of a viscous fluid with a magnetic field . Salvati, Schelling, and van Leeuwen \cite{vanLeeuwen2} numerically analyzed the evolution of the Bianchi type I universe with a viscous fluid and large-scale magnetic field. Ribeiro and Sanyal \cite{ribeirosanyal} studied a Bianchi Type $VI_{0}$ viscous fluid cosmology with an axial magnetic field in which they obtained exact solutions to the Einstein field equations assuming linear relations among the square root of matter density and the shear and expansion scalars. van Leeuwen, Miedema, and Wiersma \cite{vanLeeuwen4} proved that a non-rotating Bianchi model of class A containing a viscous fluid and magnetic field can only be of Type I or $IV_{0}$. Pradhan and Pandey \cite{pradhanpandey} studied the Bianchi Type I model with a bulk viscous fluid in addition to a varying cosmological constant. Pradhan and Singh \cite{pradhansingh} studied the Bianchi Type I model in the presence of a magnetic field and shear and bulk viscosity, but assumed that the shear tensor was proportional to the expansion tensor. Bali and Anjali \cite{balianjali} studied a Bianchi Type I magnetized fluid model with a bulk viscous string dust fluid, in which they compared their results in the presence and absence of large-scale magnetic fields. 

In this paper we examine a viscous MHD Bianchi Type I non-tilted viscous magnetohydrodynamic model. In contrast to the references cited above, which use a metric approach, we use the Hubble-normalized dynamical systems approach based upon the theory of orthonormal frames pioneered by Ellis and MacCallum \cite{ellismac}. In treating a problem with the method of Ellis and MacCallum, the Einstein field equations (a coupled set of ten hyperbolic nonlinear partial differential equations) are reduced to a system of autonomous nonlinear first-order ordinary differential equations. In a previous work \cite{isk1}, we employed such an approach to treat a Bianchi Type IV viscous model in the absence of magnetic sources. To the best of our knowledge, a treatment of a viscous MHD model along these lines has not yet appeared in the literature. In the present work, we examine the important role of the fixed points of the dynamical system. In particular we show that the fixed points may be associated with Kasner-type solutions, a flat universe FLRW solution, and interestingly, a new solution to the Einstein Field equations involving non-zero magnetic fields, and non-zero viscous coefficients. We examine several features of the dynamical system, including its early and late time asymptotic behavior, and its bifurcation behavior. Finally, numerical results are presented which illustrate the behavior of the system over long times with several initial configurations. In several cases of interest, it is shown that the dynamical model isotropizes asymptotically; that is, the spatial anisotropy and the anisotropic magnetic field decay to negligible values giving a close approximation to the present-day universe. Throughout this work, we assume that the signature of the metric tensor is $(-,+,+,+)$, and the use of \emph{geometrized units}, where $G = c = 1$.


\section{The Matter Sources}

In the absence of heat conduction, the energy-momentum tensor corresponding to a 
viscous fluid cosmological model with fluid velocity four-vector $u_a$ is given by \cite{isk1}
\begin{equation}
\label{eq:visctensor}
\mathcal{V}_{ab} = (\mu_{f} + p_{f})u_{a}u_{b} + g_{ab}p_{f} - 3\xi H h_{ab} - 2\eta \sigma_{ab},
\end{equation}
where $\mu_{f}$,  $p_{f}$, and $\sigma_{ab}$ denote the fluid's energy density, pressure, and shear tensor, respectively. In addition, the quantities $\xi$ and $\eta$ denote the bulk and shear viscosity coefficients of the fluid, respectively, $H$ denotes the Hubble parameter, and $h_{ab} \equiv u_{a}u_{b} + g_{ab}$ denotes the projection tensor corresponding to the metric signature $(-,+,+,+)$. 

The energy-momentum tensor corresponding to an electromagnetic field is given by
\cite{elliscargese}
\begin{equation}
\label{eq:emtensor}
\mathcal{T}_{ab} = \frac{1}{2}u_{a}u_{b}(E^2 + B^2) + 2u_{(a} n_{b)}^{cgd}u_{c}E_{g}B_{d} - E_{a}E_{b} - B_{a}B_{b} + \frac{1}{2}h_{ab}\left(E^2 + B^2\right),
\end{equation}
where $n^{abcd}$ is the standard skew pseudo-tensor, and $E_{a}$ and $B_{a}$ are the electric and magnetic field three-vectors, respectively. Note that in an orthonormal frame, where $g_{ab} = n_{ab} = diag(-1,1,1,1)$,  the $E^2$ and $B^2$ terms in Eq. \eqref{eq:emtensor}, take the form $E^2 \equiv E^{a}E_{a} = E_{1}^{2} + E_{2}^2 + E_{3}^2$, and $B^2 \equiv B^{a}B_{a} = B_{1}^2 + B_{2}^2 + B_{3}^2$. In this work, we assume that the cosmological model is non-tilted, and thus in both Eqs. \eqref{eq:visctensor} and \eqref{eq:emtensor} we take $u_{a}$ as the four-velocity of a comoving observer $u^{a} = (1,0,0,0)$. We also assume the ideal MHD approximation, in which case the early universe behaves as a \emph{perfect} conductor. The electric field (whose magnitude is inversely proportional to the conductivity) approaches zero, even in the presence of a non-zero electric current. In other words, we assume that after recombination, the universe is such a good conductor that the cosmic electric fields required to drive a current in it are negligible. Under these conditions, the energy-momentum tensor in Eq. \eqref{eq:emtensor} simplifies to
\begin{equation}
\label{eq:Btensor}
\mathcal{T}_{Bab} = \frac{1}{2}u_{a}u_{b}\left(B^{2}\right) - B_{a}B_{b} + \frac{1}{2}h_{ab}B^2.
\end{equation}
The total energy-momentum tensor, denoted $\mathbb{T}_{ab}$, for our cosmological model is then given by
\begin{equation}
\label{eq:emtotal}
\mathbb{T}_{ab} = \mathcal{V}_{ab} + \mathcal{T}_{Bab}.
\end{equation}

In order to formulate the evolution equations corresponding to our model, we compute from Eq. \eqref{eq:emtotal} the total energy density $\tilde{\mu}$, the total pressure $\tilde{p}$, and total anisotropic stress $\tilde{\pi}_{ab}$. Using the definitions
\begin{equation}
\tilde{\mu} = \mathbb{T}_{ab}u^{a}u^{b}, \quad \tilde{p} = \frac{1}{3}h^{ab}\mathbb{T}_{ab}, \quad \tilde{\pi}_{ab} = h^{c}_{a}h^{d}_{b}\mathbb{T}_{cd} - \tilde{p}h_{ab},
\end{equation}
we find that
\begin{equation}
\label{eq:mutotal}
\tilde{\mu} = \mu_{f} + \frac{1}{2}\left(B_{1}^2 + B_{2}^2 + B_{3}^2\right),
\end{equation}
\begin{equation}
\label{eq:ptotal}
\tilde{p}= w \mu_{f} - 3\xi H + \frac{1}{6} \left(B_{1}^2 + B_{2}^2 + B_{3}^2 \right),
\end{equation}
and
\begin{equation}
\label{eq:pitotal}
 \tilde{\pi}_{ab} = -2 \eta \sigma_{ab} - B_{a}B_{b} + \frac{1}{3}h_{ab} \left(B_{1}^2 + B_{2}^2 + B_{3}^2\right).
\end{equation}
Note that in obtaining the expression for the pressure in Eq. \eqref{eq:ptotal}, we assumed that the fluid obeys the barotropic equation of state, $p_{f} = w \mu_{f}$, where $-1 \leq w \leq 1$.

It is advantageous to re-express the above quantities as expansion-normalized variables
\cite{hewittbridsonwainwright} and we thus introduce the definitions 
\begin{equation}
\label{eq:expansource}
\tilde{\Omega} = \frac{\tilde{\mu}}{3H^2}, \quad \tilde{P} = \frac{\tilde{p}}{3H^2}, \quad \tilde{\Pi}_{ab} = \frac{\tilde{\pi}_{ab}}{H^2}.
\end{equation}
We will also define the expansion-normalized magnetic field vector as
\begin{equation}
\label{eq:bfieldnormal}
\mathcal{B}_{a} = \frac{B_{a}}{3H}.
\end{equation}
The relevant expressions for the expansion-normalized variables are then given by 
\begin{equation}
\label{eq:omegaT}
\tilde{\Omega}  = \Omega_{f} + \frac{3}{2}\left(\mathcal{B}_{1}^2 + \mathcal{B}_{2}^2 + \mathcal{B}_{3}^2\right),
\end{equation}
\begin{equation}
\label{eq:PT}
 \tilde{P} = w \Omega_{f} - 3 \xi_{0} + \frac{1}{2}\left(\mathcal{B}_{1}^2 + \mathcal{B}_{2}^2 + \mathcal{B}_{3}^2\right),
\end{equation}
and
\begin{equation}
\label{eq:PiT}
\tilde{\Pi}_{ab}= -2\eta_{0} \Sigma_{ab} - 9\mathcal{B}_{a}\mathcal{B}_{b} + 3 \delta_{ab}\left(\mathcal{B}_{1}^2 + \mathcal{B}_{2}^2 + \mathcal{B}_{3}^2\right).
\end{equation}
In Eqs. \eqref{eq:omegaT}, \eqref{eq:PT} and \eqref{eq:PiT}, $\Omega_f = \mu_f/(3H^2)$ is Hubble-normalized fluid energy density, and $\xi_{0} = \xi/(3H)$ and $\eta_{0}=\eta/(3H)$ are the expansion-normalized bulk and shear viscosity coefficients, respectively; these quantities are assumed to be \emph{non-negative constants} throughout this paper. In Eq. \eqref{eq:PiT} we also denote $\Sigma_{ab} = \sigma_{ab}/H$ as the expansion-normalized shear tensor.


\section{Bianchi Type I Universe Dynamics}

With the required energy-momentum tensor in Eq. \eqref{eq:emtotal}, and the expansion-normalized source variables (Eqs. (\ref{eq:omegaT}) - (\ref{eq:PiT})) in hand, we now derive the Bianchi Type I dynamical equations.  The general evolution equations for any Bianchi type are presented in \cite{hewittbridsonwainwright} and \cite{herviklim}. The general evolution equations in the expansion-normalized variables using our notation are:
\begin{eqnarray}
\label{eq:evolutionsys1}
\Sigma_{ij}' &=& -(2-q)\Sigma_{ij} + 2\epsilon^{km}_{(i}\Sigma_{j)k}R_{m} - \mathcal{S}_{ij} + \tilde{\Pi}_{ij} \nonumber \\
N_{ij}' &=& qN_{ij} + 2\Sigma_{(i}^{k}N_{j)k} + 2 \epsilon^{km}_{(i}N_{j)k}R_{m} \nonumber \\
A_{i}' &=& qA_{i} - \Sigma^{j}_{i}A_{j} + \epsilon_{i}^{km}A_{k} R_{m}\nonumber \\
\tilde{\Omega}' &=& (2q - 1)\tilde{\Omega} - 3\tilde{P} - \frac{1}{3}\Sigma^{j}_{i}\tilde{\Pi}^{i}_{j} + \frac{2}{3}A_{i}Q^{i} \nonumber \\
Q_{i}' &=& 2(q-1)Q_{i} - \Sigma_{i}^{j}Q_{j} - \epsilon_{i}^{km}R_{k}Q_{m} + 3A^{j}\tilde{\Pi}_{ij} + \epsilon_{i}^{km}N_{k}^{j}\tilde{\Pi}_{jm}.
\end{eqnarray}
These equations are subject to the constraints
\begin{eqnarray}
\label{eq:constraints1}
N_{i}^{j}A_{j} &=& 0, \nonumber \\
\tilde{\Omega} &=& 1 - \Sigma^2 - K, \nonumber \\
Q_{i} &=& 3\Sigma_{i}^{k} A_{k} - \epsilon_{i}^{km}\Sigma^{j}_{k}N_{jm}.
\end{eqnarray}
As in Eq. \eqref{eq:expansource}, we have made use of the following notation:
\begin{equation}
\label{eq:notation1}
\left(\Sigma_{ij}, R^{i}, N^{ij}, A_{i}\right) = \frac{1}{H}\left(\sigma_{ij}, \Omega^{i}, n^{ij}, a_{i}\right) , \quad \left(\tilde{\Omega}, \tilde{P}, Q_{i}, \tilde{\Pi}_{ij}\right) = \frac{1}{3H^2}\left(\tilde{\mu}, \tilde{p}, q_{i}, \tilde{\pi}_{ij}\right).
\end{equation}
In the expansion-normalized approach, the kinematic shear tensor $\Sigma_{ab}$ describes the anisotropy in the Hubble flow, $A_{i}$ and $N^{ij}$ describe the spatial curvature, while $\Omega^{i}$ describes the relative orientation of the shear and spatial curvature eigenframes. The Bianchi Type I model is a flat anisotropic model and is Abelian, and therefore has the property that
\begin{equation}
\label{eq:bianchi1}
A^{i} = 0, \quad N_{11} = N_{22} = N_{33} = 0.
\end{equation}

The dynamical system \eqref{eq:evolutionsys1} evolves according to a dimensionless time variable, $\tau$ such that
\begin{equation}
\label{eq:Hdef1}
\frac{dt}{d\tau} = \frac{1}{H},
\end{equation}
where $H$ is the Hubble parameter with evolution equation 
\begin{equation}
\label{eq:Hdef2}
H' = -(1 + q) H.
\end{equation}
The deceleration parameter $q$ is very important in the expansion-normalized approach:
when $q<-1$ the universe expansion is accelerating, when $q>-1$ the universe expansion is decelerating, and when $q=-1$ the universe is static, that is, it is not self-similar. From Eq. (1.90) in \cite{ellis}, and using Eq. \eqref{eq:notation1}, the parameter $q$ may be written as 
\begin{eqnarray}
\label{eq:qdef}
q &\equiv& 2\Sigma^2 + \frac{1}{2}\left(\tilde{\Omega} + 3\tilde{P}\right) \nonumber \\
&=& 2\Sigma^{2} + \Omega_{f} \left(\frac{1}{2} + \frac{3w}{2}\right) - \frac{9}{2}\xi_{0} + \frac{3}{2}\left(\mathcal{B}_{1}^2 + \mathcal{B}_{2}^2 + \mathcal{B}_{3}^2\right), 
\end{eqnarray}
where $2\Sigma^{2} \equiv  \left(\Sigma_{ab} \Sigma^{ab}\right)/3$. 

In the case of a magnetic field source, one must also include an evolution equation for the magnetic field, which is the orthonormal frame analog of the standard Maxwell-Faraday equation. According to Eq. (71) in \cite{vanelstuggla}, Eq. (2.4) in \cite{leblanc3}, and Eqs. \eqref{eq:bfieldnormal}, \eqref{eq:notation1}, \eqref{eq:Hdef1}, and \eqref{eq:Hdef2} above, the magnetic field evolution is given by
\begin{equation}
\label{eq:bfieldevolution}
\mathcal{B}_{a}' = \mathcal{B}_{a}\left(-1+q\right) + \Sigma_{ab}\mathcal{B}^{b} + \epsilon_{abv}R^{v}\mathcal{B}^{b}.
\end{equation}

For convenience, we introduce the notation
\begin{equation}
\label{eq:sheardefs}
\Sigma_{+} = \frac{1}{2} \left(\Sigma_{22} + \Sigma_{33}\right), \quad \Sigma_{-} = \frac{1}{2\sqrt{3}}\left(\Sigma_{22} - \Sigma_{33}\right),
\end{equation}
such that $\Sigma^{2} = \Sigma_{+}^2 + \Sigma_{-}^2$.
In the evolution equations \eqref{eq:evolutionsys1}, the expansion-normalized angular velocity variables $R_{a}$ can be found from the non-diagonal shear equations, $\Sigma_{12}', \Sigma_{23}'$, and $\Sigma_{13}'$. From these equations, we get that
\begin{equation}
\label{eq:Rdefs}
R_{1} = -\frac{3\sqrt{3}\mathcal{B}_{2} \mathcal{B}_{3}}{2\Sigma_{-}}, \quad R_{2} = \frac{9\mathcal{B}_{1} \mathcal{B}_{3}}{\sqrt{3}\Sigma_{-} - 3\Sigma_{+}}, \quad R_{3} = \frac{9\mathcal{B}_{1} \mathcal{B}_{2}}{\sqrt{3}\Sigma_{-} + 3\Sigma_{+}}.
\end{equation}
To avoid situations where $R_{1}$, $R_{2}$, or $R_{3}$ become singular, we will set $\mathcal{B}_{1} = \mathcal{B}_{3} = 0$, and keep $\mathcal{B}_{2} \neq 0$, hence assuming that the magnetic field acts in a single spatial direction,  as is done in \cite{tsagasmaartens}, \cite{doroshkevich}, \cite{collins}, and \cite{thorne2}. 
Then, $R_{1} = R_{2} = R_{3} = 0$, and according to Eqs. \eqref{eq:Rdefs}, \eqref{eq:sheardefs}, \eqref{eq:bianchi1}, and \eqref{eq:qdef}, the evolution equations \eqref{eq:evolutionsys1} become: 
\begin{eqnarray}
\label{eq1:b1evolutionsys}
\Sigma_{+}' &=& -\frac{3}{2}\mathcal{B}_{2}^2 + \Sigma_{+}\left[q - 2\left(1 + \eta_{0} \right)\right], \\
\label{eq2:b1evolutionsys}
\Sigma_{-}'  &=& -\frac{3 \sqrt{3}}{2} \mathcal{B}_{2}^2 + \Sigma_{-}\left[q - 2\left(1 + \eta_{0} \right)\right], \\
\label{eq3:b1evolutionsys}
\mathcal{B}_{2}' &=& \mathcal{B}_{2} \left(-1 + q + \sqrt{3}\Sigma_{-} + \Sigma_{+}\right),
\end{eqnarray}
where the deceleration parameter is now given by
\begin{equation}
\label{eq:qdef2}
q = 2\left(\Sigma_{+}^2 + \Sigma_{-}^2\right)+ \Omega_{f} \left(\frac{1}{2} + \frac{3w}{2}\right) - \frac{9}{2}\xi_{0} + \frac{3}{2}\mathcal{B}_{2}^2.
\end{equation}
In Eq. \eqref{eq:qdef2} we have defined the energy density as
\begin{equation}
\label{eq4:b1evolutionsys}
\Omega_{f} = 1 - \frac{3}{2} \mathcal{B}_{2}^2 - \Sigma_{-}^2 - \Sigma_{+}^2 \geq 0,
\end{equation}
which, as indicated in Eq. \eqref{eq4:b1evolutionsys}, is restricted to be non-negative on physical grounds. After some algebra, the auxiliary equation in \eqref{eq:evolutionsys1} becomes
\begin{equation}
\label{eq:OmegafP}
\Omega_{f}' = \Omega_{f} \left(2q-1-3w\right) + 4 \eta_{0} \left(\Sigma_{+}^{2} + \Sigma_{-}^2\right) + 9\xi_{0}.
\end{equation}
In seeking solutions to \eqref{eq1:b1evolutionsys}, \eqref{eq2:b1evolutionsys} and 
\eqref{eq3:b1evolutionsys}, we further enforce the physical restrictions
\begin{equation}
\label{eq:restrict1}
-1\leq w \leq 1, \quad \xi_{0} \geq 0, \quad \eta_{0} \geq 0,
\end{equation}
on the state parameter, bulk and shear viscosity coefficients, respectively. Any combinations of these parameters must additionally satisfy $\Omega_{f} \geq 0$, $\Sigma_{+} \in \mathbb{R}, \Sigma_{-} \in \mathbb{R}$, and $\mathcal{B}_{2} \geq 0 \in \mathbb{R}$.


\section{A Fixed Point Analysis}

We now consider the local stability of the equilibrium points of the system \eqref{eq1:b1evolutionsys}-\eqref{eq3:b1evolutionsys}, which we abbreviate as
\begin{equation}
\label{eq:basedef}
\mathbf{x}' = \mathbf{f(x)}.
\end{equation}
Here $\mathbf{x} = \left[\Sigma_{+}, \Sigma_{-}, \mathcal{B}_{2}\right] \in \mathbf{R}^{3}$, and  the vector function $\mathbf{f(x)}$ denotes the right-hand-side of the dynamical system. The state space of the system is the subset of $\mathbb{R}^{3}$ defined by the inequality in Eq. \eqref{eq4:b1evolutionsys}, which is equivalent to
\begin{equation}
\label{eq:ineq1}
\Sigma_{+}^2 + \Sigma_{-}^2 + \frac{3}{2}\mathcal{B}_{2}^2 \leq 1,
\end{equation}
so the state space is clearly bounded. This inequality also is a constraint for the initial conditions of the dynamical system. There is only one symmetry of the dynamical system, given by
\begin{equation}
\label{eq:symm1}
\left[\Sigma_{+}, \Sigma_{-}, \mathcal{B}_{2}\right] \rightarrow \left[\Sigma_{+}, \Sigma_{-}, -\mathcal{B}_{2}\right].
\end{equation}
The system is therefore invariant with respect to spatial inversions in the function $\mathcal{B}_{2}$, and we can take $\mathcal{B}_{2} \geq 0$. In most cases, we examine the stability of the critical points $\mathbf{a}$ where $\mathbf{f}(\mathbf{a}) = 0$ by locally linearizing the system leading to the relationship $\mathbf{x}' = D \mathbf{f(a)} \mathbf{x}$. The stability of the system is then determined the sign of the eigenvalues of the Jacobian matrix $D \mathbf{f(a)}$. In the work that follows, we will denote eigenvalues of the dynamical system by $\lambda_{i}$, where $i = 1,2,3,...$.

\subsection{Kasner Equilibrium Points}

We now discuss a set of equilibrium points which are known as the Kasner solutions to the system \cite{ellis}. Each such equilibrium point corresponds to a vacuum solution and is unstable for our model. These equilibrium points, the set of which we denote ${\mathcal{K}}$, lie on the {\em Kasner circle}
\begin{equation}
\Sigma_{-}^2 + \Sigma_{+}^2 = 1
\end{equation}
in the plane $\mathcal{B}_{2} = 0$ for parameter values $\xi_{0} = \eta_{0} = 0$, and $-1 \leq w \leq 1$. The cosmological parameters at every point on the Kasner circle are
\begin{equation}
\Omega_{f} = 0, \quad q = 2, \quad \Sigma^2 = 1.
\end{equation}
The eigenvalues of the Jacobian matrix at each point are
\begin{equation}
\label{eq:Kasner_lambda}
\lambda_{1} = 0, \quad \lambda_{2} = 3(1-w), \quad \lambda_{3} = 1 + \Sigma_{+} - \sqrt{3(1-\Sigma_{+}^2)}.
\end{equation}
As can be seen from Eq. \eqref{eq:Kasner_lambda} when $w = 1$ two of the eigenvalues are zero, and these equilibrium points are not normally hyperbolic. One can therefore not use linearization methods to determine the local asymptotic behavior. In the following discussion we restrict our attention to the parameter region defined by $-1 \leq w < 1$.

Let us parametrize the Kasner circle points using the polar angle $\psi$ as is done in \cite{ellis}:
\begin{equation}
\Sigma_{+} = \cos \psi, \quad \Sigma_{-} = \sin \psi, \quad -\pi < \psi \leq \pi.
\end{equation}
The Kasner exponents $p_1$, $p_2$, and $p_3$ of the Kasner metric
\begin{equation}
ds^2 = -dt^2 + t^{2p_1} dx^2 + t^{2p_2} dy^2 + t^{2p_3} dz^2
\end{equation}
are then given by
\begin{equation}
p_{1} = \frac{1}{3} \left(1- 2 \cos \psi\right), \quad p_{2,3} = \frac{1}{3} \left(1 + \cos \psi \pm \sqrt{3} \sin \psi \right).
\end{equation}
It is well known that the Taub points occur for $\psi=-\pi/3$, $\pi$, and  $\pi/3$. We use these Taub points to subdivide the circle $\mathcal{K}$ into three open arcs. Along the arc $\mathcal{K}_1$ defined by
\begin{equation}
-\frac{\pi}{3} < \psi < \frac{\pi}{3}
\end{equation}
the eigenvalue $\lambda_3$ is positive, and hence each point on the arc corresponds to a source.
Furthermore, on $\mathcal{K}_{1}$  we have $p_{1} < 0$, $p_{2} > 0$, and $p_{3} > 0$ which implies that each of these equilibrium points represent a cigar-type past singularity of the system.
Along the arcs $\mathcal{K}_2$ and $\mathcal{K}_3$ defined by
\begin{equation}
\quad -\pi < \psi < -\frac{\pi}{3}
\quad {\mbox{and}} \quad \frac{\pi}{3} < \psi < \pi,
\end{equation}
respectively, the eigenvalue $\lambda_3$ is negative and each Kasner point on these arcs corresponds to a local saddle point. On both these arcs we also have $p_{1} > 0$, $p_{2} > 0$, and $p_{3} < 0$
which corresponds to a cigar-type singularity as well. In the case of a cigar singularity, matter collapses in along one spatial direction from infinity, halts, and then begins to re-expand, while in the other spatial directions, the matter expands monotonically at all times. Each Taub point, on the other hand, corresponds to a pancake singularity, where matter is found to expand monotonically in all directions, starting from a very high expansion rate in one spatial direction, but from zero expansion rates in the other spatial directions (Page 144, \cite{ellis3}).

\subsection{Flat Universe Equilibrium Point}
This equilibrium point, which we denote as $\mathcal{F}$, occurs for  
\begin{eqnarray}
\label{eq:point2}
\Sigma_{+} = 0, \quad \Sigma_{-} = 0, \quad \mathcal{B}_{2} = 0,
\end{eqnarray}
and represents the \emph{flat FLRW} universe. The cosmological parameters at this point take the form
\begin{equation}
\label{eq:params2}
\Omega_{f} = 1, \quad q = \frac{1}{2}\left(1 + 3w - 9\xi_{0}\right), \quad \Sigma^2 = 0.
\end{equation}
The eigenvalues of the Jacobian matrix of the dynamical system at $\mathcal{F}$ are given by
\begin{equation}
\label{eq:flrweigs}
\lambda_{1} = \frac{1}{2} (-1+3 w-9 \xi_{0} ), \quad \lambda_{2} = \lambda_{3} = \frac{1}{2} (-3+3 w-4 \eta_{0} -9 \xi_{0} ),
\end{equation}
where in Eqs. \eqref{eq:params2} and \eqref{eq:flrweigs} we require that $\eta_{0} \geq 0$, $\xi_{0} \geq 0$, and $-1 \leq w \leq 1$.

The point $\mathcal{F}$ represents a local sink if
\begin{equation}
\label{eq:fsink1}
\eta_{0} \geq 0, \quad \xi_{0} \geq 0, \quad -1 \leq w < \frac{1}{3}, 
\end{equation}
or
\begin{equation}
\label{eq:fsink2}
\eta_{0} \geq 0, \quad \frac{1}{3} \leq w \leq 1, \quad \xi_{0} > \frac{1}{9} \left(-1 + 3w\right).
\end{equation}
In Fig. (\ref{fig:fig1}), we have denoted the region defined by \eqref{eq:fsink1} and \eqref{eq:fsink2} as $S1(F)$.

The point $\mathcal{F}$ represents a saddle point if
\begin{equation}
\label{eq:fsaddle1}
\eta_{0} = 0, \quad \frac{1}{3} < w < 1, \quad 0 \leq \xi_{0} < \frac{1}{9}(-1+3w),
\end{equation}
or
\begin{equation}
\label{eq:fsaddle2}
\eta_{0} = 0, \quad w = 1, \quad 0 < \xi_{0} < \frac{2}{9},
\end{equation}
or
\begin{equation}
\label{eq:fsaddle3}
\eta_{0} > 0, \quad \frac{1}{3} < w \leq 1, \quad 0 \leq \xi_{0} < \frac{1}{9}\left(-1+ 3w\right),
\end{equation}
where in each case $\lambda_{1} > 0$ and $\lambda_{2} = \lambda_{3} < 0$.
We will subsequently denote the region defined by \eqref{eq:fsaddle1} - \eqref{eq:fsaddle3} as SA(F).

The point $\mathcal{F}$ can also represent a local source if
\begin{equation}
\label{eq:fsource1}
\eta_{0} = 0, \quad w = 1, \quad \xi_{0} = 0,
\end{equation}
where in this case, $\lambda_{1} > 0$ and $\lambda_{2} = \lambda_{3} = 0$. An analysis nearly identical to that presented in the classification of the Kasner point $\mathcal{K}_1$ does confirm this is a source point. We will subsequently denote the region defined by Eq. \eqref{eq:fsource1} as U(F).

It is important to note that $q = -1$ when $0 \leq \xi_{0} \leq \frac{2}{3}$ and $w = 3\xi_{0} - 1$, and thus the equilibrium point in the domain defined by these values of $\eta_{0}$, $\xi_{0}$ , and $w$ does not correspond to a self-similar solution. In particular, if one chooses $\xi_{0} = 0$ such that $w = -1$, the corresponding model is locally the de Sitter solution \cite{elliscosmo}.

\subsection{A New Equilibrium Point}
We will denote this equilibrium point as $\mathcal{BI_{MV}}$. For brevity in our presentation, we introduce the condensed notation for the fixed points
\begin{equation}
\label{eq:p3points}
\Sigma_{+} = -\frac{1}{16\alpha} (\beta_1 + \gamma),\quad\Sigma_{-} = -\frac{\sqrt{3}}{16\alpha} (\beta_1 + \gamma), \quad \mathcal{B}_{2} = \frac{1}{4 \sqrt{3}} \left(\beta_2 - \gamma\right)^{1/2},
\end{equation}
where
\begin{equation}
\label{eq:alpha}
\alpha = 5-6 \eta_{0} +3w (1+2 \eta_{0} ),
\end{equation}
\begin{equation}
\label{eq:beta1}
\beta_1 = 9 w^2 (1+2\eta_{0} )^2 + 12w(1+2\eta_{0})(3-2\eta_{0}) +
(7-2\eta_{0})(5-6\eta_{0}),
\end{equation}
\begin{equation}
\label{eq:beta2}
\beta_2 = -9 w^2 (1+2\eta_{0} )^2 + 12w(1-2\eta_0)^2 -
(17-6\eta_0)(3-2\eta_0) -144 \xi_{0},
\end{equation}
and
\begin{equation}
\label{eq:gamma}
\gamma = |\alpha|\left[9w^2(1+2\eta_0)^2 - 6w(3-2\eta_0)^2 + (7-2\eta_0)^2
+ 32(1+9\xi_0)\right]^{1/2}.
\end{equation}
Similarly, the cosmological parameters at this point take the form
\begin{equation}
\Omega_{f} = -\frac{1}{16\alpha}\left(\beta_3 + (1+2\eta_0)\gamma\right),
\quad q = \frac{1}{4\alpha} (\beta_4 + \gamma),\quad
\Sigma^2 = \frac{1}{64\alpha^2} \left(\beta_1 + \gamma\right)^2,
\end{equation}
where
\begin{equation}
\label{eq:beta3}
\beta_3 = 9 w^2 (1+2 \eta_{0} )^3 -12 w (1-2 \eta_{0} )^2 (1+2 \eta_{0} )
-(5-6\eta_0)(3-2\eta_0)^2
\end{equation}
and
\begin{equation}
\label{eq:beta4}
\beta_4 = 9 w^2 (1+2 \eta_{0} )^2 +24 w (1+2 \eta_{0} ) (2- \eta_{0} )
+(5-6\eta_0)(11-2\eta_0).
\end{equation}
The restrictions require us to set
\begin{equation}
\label{eq:p3restr2}
\eta_{0} > \frac{3}{2}, \quad  \frac{1}{3} \leq w < \frac{-5+6\eta_{0}}{3+ 6 \eta_{0}}, \quad 0 \leq \xi_{0} \leq \frac{1}{9}\left(-1 + 3w\right).
\end{equation}
We will subsequently denote the parameter region defined by \eqref{eq:p3restr2} as S2(BIMV). We were not able to obtain exact expressions for the eigenvalues in this region due to overwhelming algebraic complexity; however, comprehensive numerical experiments demonstrate that the eigenvalues in this region are either zero or negative thus corresponding to a sink. Interestingly, for a fixed value of $\eta_0 > 3/2$, the dependence of the largest eigenvalue $\lambda_1$ on the parameters $w$ and $\xi_0$ is very nearly linear on S2(BIMV). For several values of $\eta_0$ a planar approximation for the $\lambda_1$ surface was constructed in our numerical experiments using computed values in the $(w,\xi_0)$ domain. The planar approximation with equation $\lambda_1 = 1-3w+9\xi_0$ agreed with numerically-computed values of $\lambda_1$  everywhere in S2(BIMV) to within 3 to 5 digits accuracy, depending on the value of $\eta_0 $ chosen in the range $3/2 < \eta_0 \leq 500 $; the best agreement was obtained for larger values of $\eta_0$. Despite the algebraic complexity, we can show analytically that the equilibrium point corresponding to parameters in S2(BIMV) is indeed a sink by the following considerations. For convenience we have included the Jacobian matrix $D \mathbf{f(a)}$ (where ${\bf{a}}$ is the equilibrium point under consideration) in Appendix A. As we discuss in the following section, the surface
$\xi_0 =(3w-1)/9$, which forms one boundary of the domain S2(BIMV), corresponds to bifurcations in the solution; on this surface the Jacobian matrix is diagonal and its eigenvalues are seen to be 
\begin{equation}
\lambda_{1} = 0, \lambda_{2} = \lambda_3= -1 - 2\eta_{0}.
\end{equation}
We seek to characterize the equilibrium point slightly inside the region S2(BIMV), and thus in what follows we find expressions for the eigenvalues corresponding to $\xi_0 =(3w-1)/9 - \varepsilon$, where $\varepsilon > 0$ is a small parameter chosen to ensure that indeed $\xi_0\geq 0$ and $w<(6\eta_0-5)/(6\eta_0+3)$. Expanding the elements of the Jacobian matrix in a series in $\varepsilon$ to first order allows simple expressions for the eigenvalues to be obtained:
\begin{equation}
\label{eq:approx_eigs}
\lambda_{1} = -9\varepsilon, \lambda_{2} = -1 - 2\eta_{0} + \frac{108\varepsilon}{(7-2\eta_0)+3w(1+2\eta_0)},
\lambda_{3} = -1 - 2\eta_{0} + \frac{36\varepsilon}{(7-2\eta_0)+3w(1+2\eta_0)}.
\end{equation}
We note that all the terms in \eqref{eq:approx_eigs} have error of order ${\cal{O}}(\varepsilon^2)$. The quantity $\varepsilon$ may be taken arbitrarily small, and thus all the eigenvalues corresponding to parameters slightly inside the bifurcation boundary are negative; i.e., the equilibrium point is a local sink. Since the solution does not bifurcate inside the region S2(BIMV) -- it does so only across its boundaries -- it follows that {\em all} parameter values inside the region correspond to a local sink. In addition, the results 
\eqref{eq:approx_eigs} indicate that $\partial \lambda_1 /\partial \xi_0 \approx 9$ at the boundary $\xi_0 =(3w-1)/9$; this approximation for the $\xi_0$-slope of the $\lambda_1$ surface agreed to several digits with the same quantity which was numerically computed and used to form the planar approximation for this surface discussed above.

To the best knowledge of the authors the equilibrium point $\mathcal{BI_{MV}}$ has not previously been reported in the literature, and represents a new solution to the Einstein Field Equations. Interestingly, the model with parameter values in S2(BIMV) will not isotropize, since this equilibrium point is a local source with $\Sigma_{+}, \Sigma_{-}, \mathcal{B}_{2} \neq 0$.

For convenience, we have summarized the results of this section in Fig. (\ref{fig:fig1}) which depicts the different regions of sinks, saddles, and sources of the dynamical system.
\begin{figure}[H]
\begin{center}
\caption{A depiction of the different regions of sinks, sources, and saddles of the dynamical system as defined by the aforementioned restrictions on the values of the expansion-normalized bulk and shear viscosities, $\xi_{0}, \eta_{0}$ and equation of state parameter, $w$.}
\includegraphics[scale = 0.75]{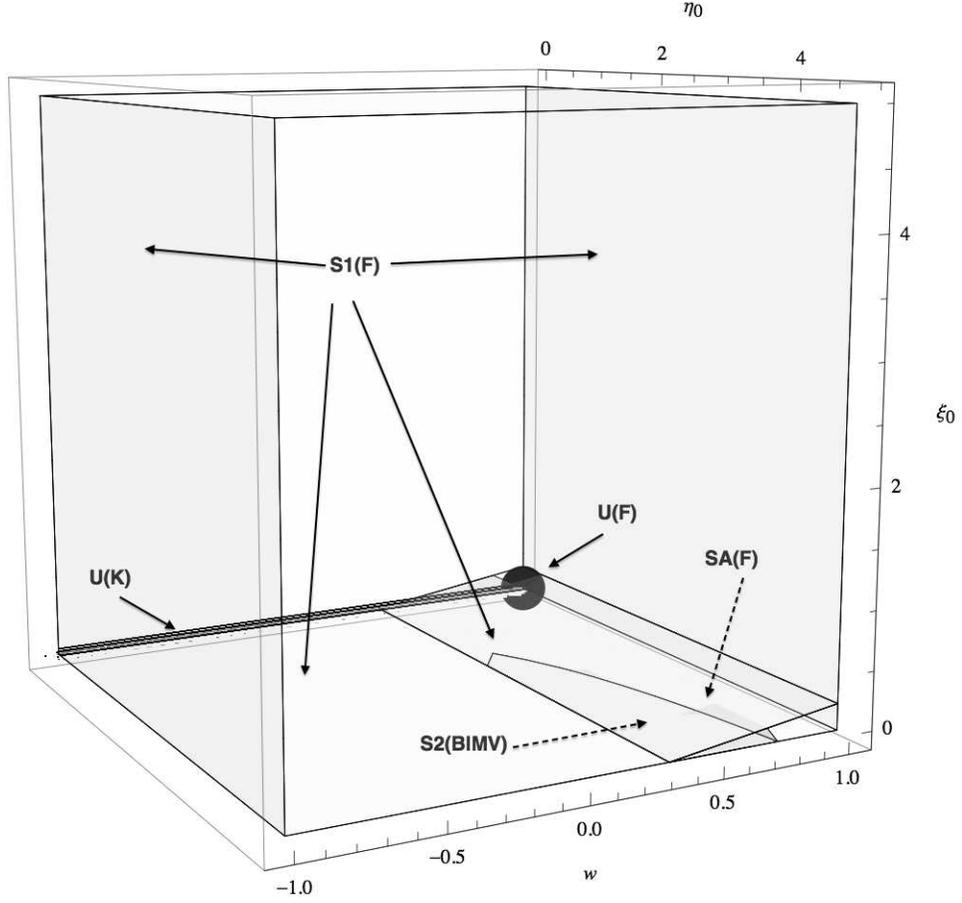}
\label{fig:fig1}
\end{center}
\end{figure}

\newpage

\section{Bifurcations}
We note that the equilibria found above are related to each other by sequences of bifurcations. We will now give in this section the details of these bifurcations. The method we use involves determining for what values of $\eta_{0}$, $\xi_{0}$, and $w$ do the different equilibrium points destabilize. A similar method was employed in \cite{hewittbridsonwainwright}.

The linearized system for points on $\mathcal{K}$, where $\Sigma_{-} = \pm\sqrt{1 - \Sigma_{+}^2}$ becomes:
\begin{eqnarray}
\Sigma_{+}' &=& -3(-1+w)\Sigma_{+}^3 -3(-1+w)\Sigma_{+}^2\sqrt{1-\Sigma_{+}^2}, \\
\Sigma_{-}' &=&  -3(-1+w)\Sigma_{+}^2 \sqrt{1-\Sigma_{+}^2} + 3\Sigma_{-}(-1+w)(-1+\Sigma_{+}^2),\\
\mathcal{B}_{2}' &=& (1+\Sigma_{+} + \sqrt{3(1-\Sigma_{+}^2)}.
\end{eqnarray}
We can therefore see that $\Sigma_{+}$ destabilizes $\mathcal{K}$ when $\Sigma_{+} = 0, -1 \leq w < 1$, $\Sigma_{-}$ destabilizes $\mathcal{K}$ when $\Sigma_{+} = \pm 1, -1 \leq w < 1$, and $\mathcal{B}_{2}$ destabilizes $\mathcal{K}$ when $\Sigma_{+} = -1, -1 \leq w < 1$, where in each case $\xi_{0} = \eta_{0} = 0$.

We next consider the linearized system at $\mathcal{F}$, given by
\begin{eqnarray}
\Sigma_{+}' &=& \frac{1}{2} \left(-3 + 3w - 4 \eta_{0} - 9\xi_{0}\right)\Sigma_{+}, \\
\Sigma_{-}' &=& \frac{1}{2} \left(-3 + 3w - 4 \eta_{0} - 9\xi_{0}\right)\Sigma_{-}, \\
\mathcal{B}_{2}' &=& \frac{1}{2}\left(-1 + 3w - 9 \xi_{0}\right)\mathcal{B}_{2}.
\end{eqnarray}
It may be seen that both $\Sigma_{+}$ and $\Sigma_{-}$ destabilize $\mathcal{F}$ when $\xi_{0} = \eta_{0} = 0$, and $w = 1$, while $\mathcal{B}_{2}$ destabilizes $\mathcal{F}$ when $\eta_{0} \geq 0$, $1/3 \leq w \leq 1$, and $\xi_{0} = \left(3w-1\right)/9$.

We now turn to the final equilibrium point of the system, $\mathcal{BI_{MV}}$, whose corresponding Jacobian matrix is given in Appendix A. It may be seen that the Jacobian is in fact diagonal when $\eta_{0} > 3/2$ and $\xi_0 = (3w-1)/9$ in which case $\mathcal{B}_{2}'=0$. Thus $\mathcal{B}_{2}$ destabilizes this equilibrium point along the surface $\xi_0 = (3w-1)/9$, which is a shared boundary with the region SA(F). Across this boundary, the source point in SA(F) becomes a sink in S2(BIMV). Extensive numerical experiments indicated that there were no other destabilizations for this equilibrium point.

We thus see that the system destabilizes either on the line in parameter space $\xi_0=\eta_0=0$ or it destabilizes on the parameter surface $\xi_0 = (3w-1)/9$. Given this information on the destabilizations, we see that some possible bifurcation sequences can be obtained as follows. First, let us set $w = 1/3, \xi_{0} = 0$. Then, we have that:
\begin{equation}
\mathcal{K}_{(\eta_{0} = 0)} \rightarrow \mathcal{F}_{(0 < \eta_{0} \leq 3/2)} \rightarrow \mathcal{BI_{MV}}_{(\eta_{0} > 3/2)}. 
\end{equation}
Another possible bifurcation sequence is obtained when:
\begin{equation}
\mathcal{K}_{(-1 \leq w < 1/3, \xi_{0} = \eta_{0}=0)} \rightarrow \mathcal{BI_{MV}}_{(\eta_{0} > 3/2, w = 1/3, \xi_{0} = 0)} \rightarrow \mathcal{F}_{(\eta_{0} > 3/2, 1/3 < w \leq 1, 0 < \xi_{0} \leq 2/9)}.
\end{equation}
One can also have that
\begin{equation}
\mathcal{K}_{(-1 \leq w < 1/3, \xi_{0} = \eta_{0} =0)} \rightarrow \mathcal{BI_{MV}}_{(\eta_{0} > 3/2, w = 1/3, \xi_{0} = 0)} \rightarrow \mathcal{F}_{(\eta_{0} = \xi_{0} = 0, w = 1)}.
\end{equation}

As discussed previously, the surface $\xi_{0} = \left(3w-1\right)/9$ governs bifurcations of the dynamical system. It is constructive to display this bifurcation behavior for cases where first $\xi_{0} < \left(3w-1\right)/9$, then $\xi_{0} = \left(3w-1\right)/9$, and finally, $\xi_{0} > \left(3w-1\right)/9$. For the purposes of this numerical experiment, we specifically chose $w = 1/2$, $\eta_{0} = 5$, therefore requiring that the three aforementioned cases reduce to $\xi_{0} < 1/18$, $\xi_{0} = 1/18$, and $\xi_{0} > 1/18$. The outcomes of this experiment are shown in Fig. (2).
\newpage
\begin{figure}[H]
\label{fig:fig10}
\caption{These figures show bifurcation behavior for a varying expansion-normalized bulk viscosity coefficient, $\xi_{0}$, while $w$ and $\eta_{0}$ were held fixed. The circles indicate the $\mathcal{BI_{MV}}$ equilibrium points, while the diamond indicates the FLRW equilibrium point. For the first figure, $\xi_{0} =  0.05$, for the second figure, $\xi_{0} = 1/18$, and for the last figure, $\xi_{0} = 0.6$. Note how the increasing values of $\xi_{0}$ first result in a slight shift of the equilibrium point position of $\mathcal{BI_{MV}}$, and then finally a transition to a new state, namely the FLRW equilibrium, which was predicted by our fixed-point analysis.}
\includegraphics*[scale = 0.50]{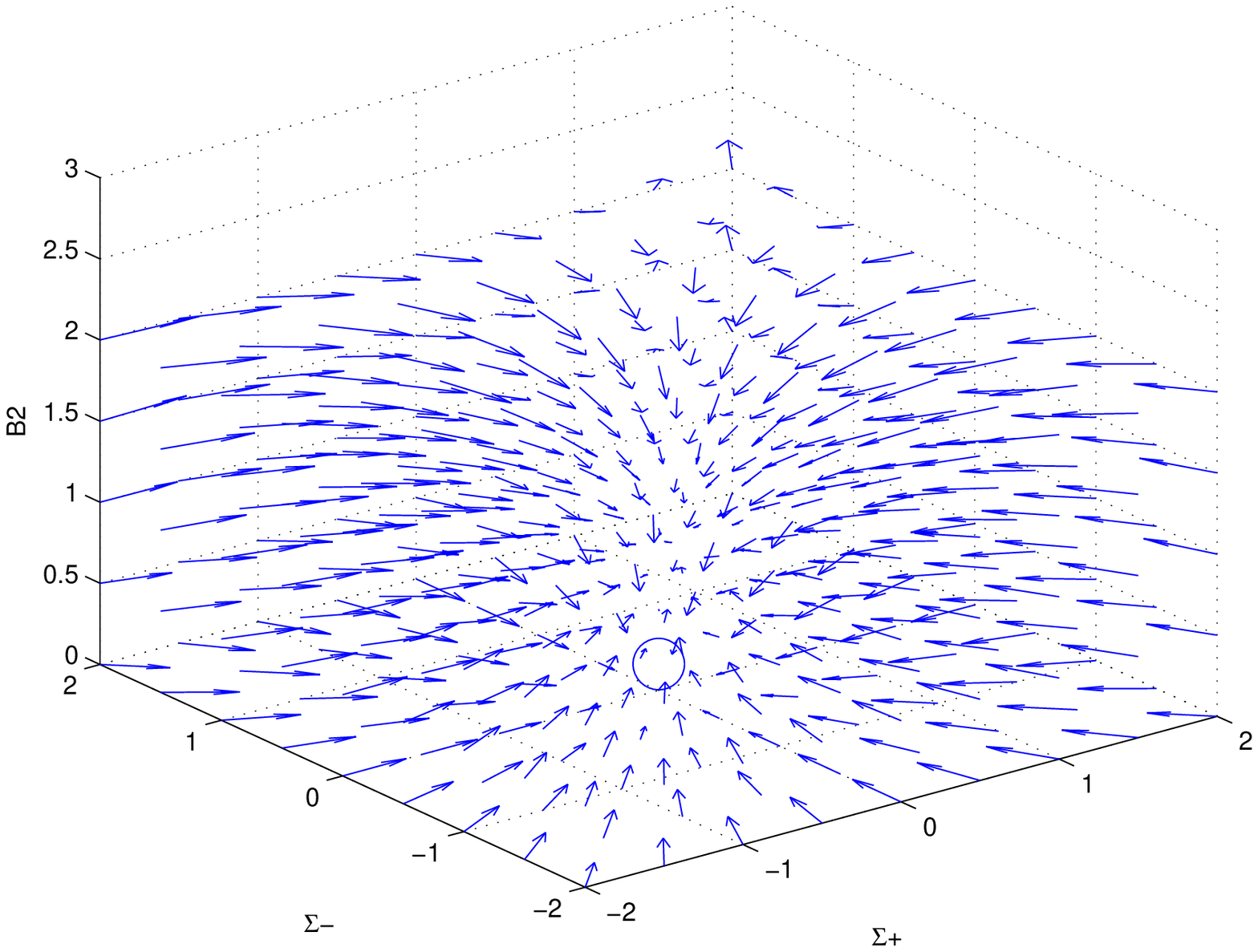} \\
\includegraphics*[scale = 0.50]{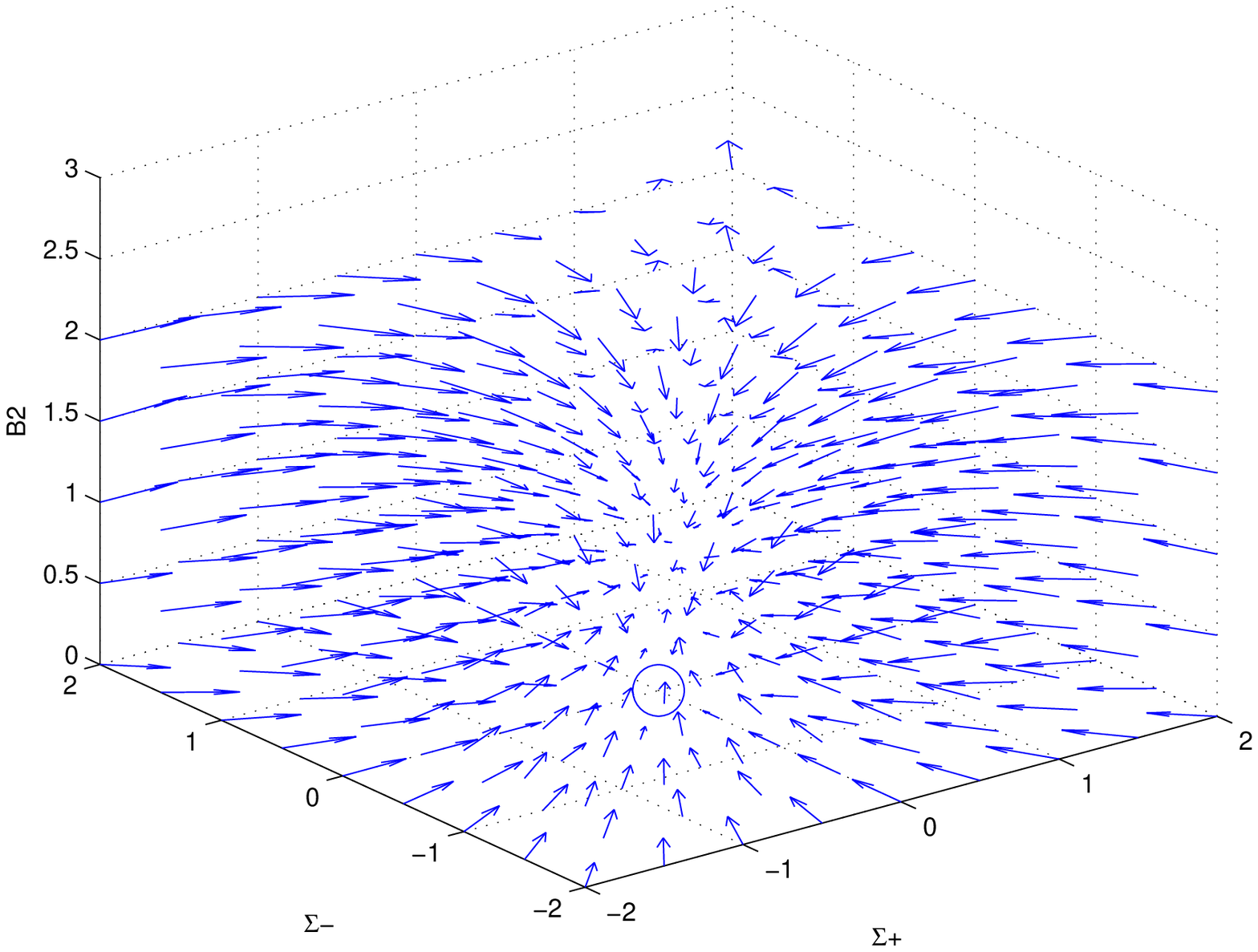} 
\includegraphics*[scale = 0.50]{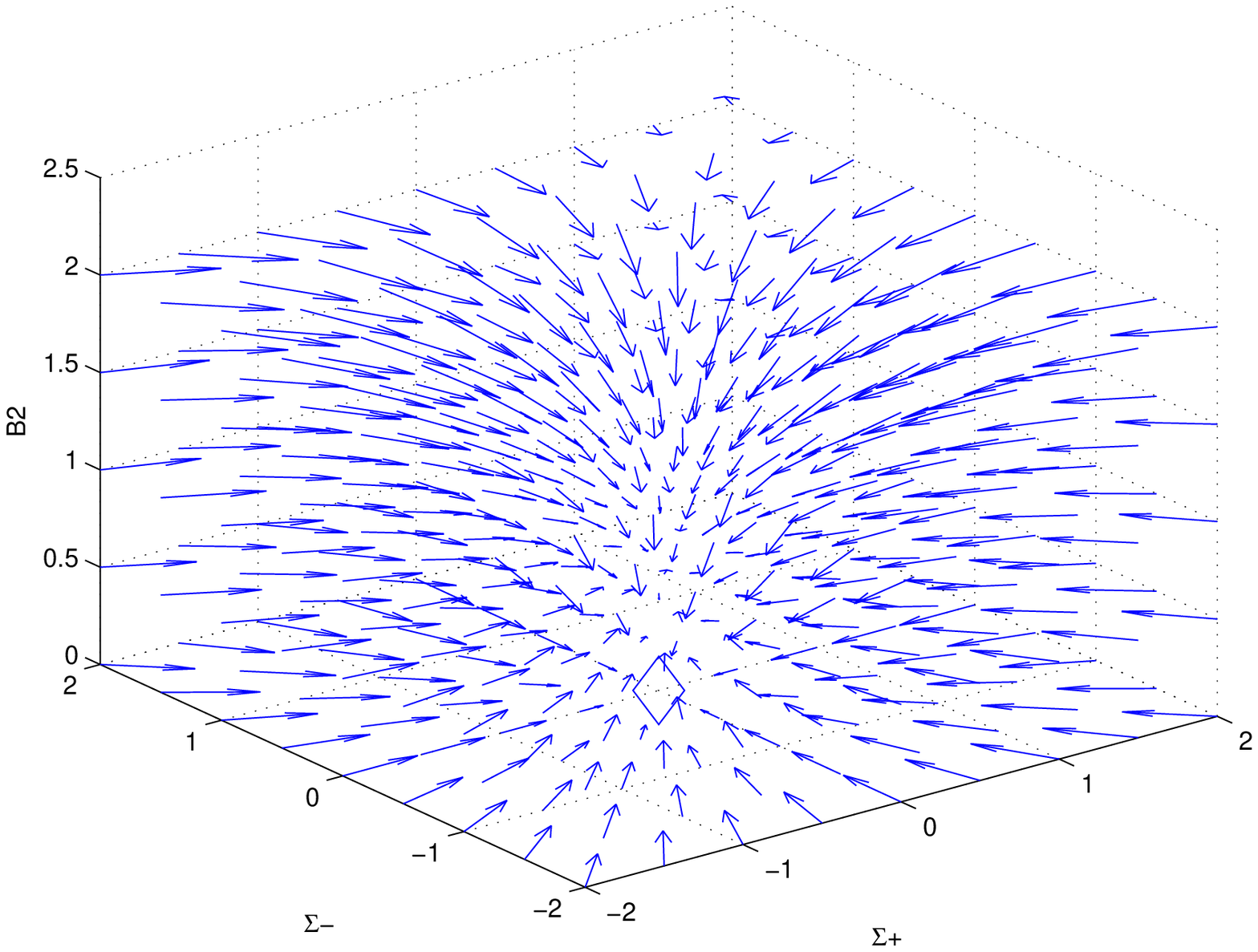} 
\end{figure}
\newpage


\section{Qualitative Properties of the System}

\subsection{A Further Analysis of the Asymptotic Behavior}
An important question to ask in analyzing some qualitative properties of the dynamical system is whether there are any invariant sets of the dynamical system. A very useful proposition in this regard is given by Proposition 4.1 in reference \cite{ellis}, which states that for a dynamical system of type \eqref{eq:basedef}, if $Z: \mathbb{R}^{n} \to \mathbb{R}$ is a $C^{1}$ function such that $Z' = \alpha Z$, where $\alpha: \mathbb{R}^{n} \to \mathbb{R}$ is a continuous function, then the subsets of $\mathbb{R}^{n}$ defined by $Z > 0$, $Z = 0$, and $Z < 0$ are invariant sets of the flow of the system of differential equations. From Eq. \eqref{eq3:b1evolutionsys}, we see that this proposition applies with $Z = \mathcal{B}_{2}$, and thus $\mathcal{B}_{2} = 0$ and $\mathcal{B}_{2} > 0$ are invariant sets of the system. We also note that if one sets $\eta_{0} = \xi_{0} = 0$ in Eq. \eqref{eq:OmegafP}, then the proposition also applies with $Z = \Omega_{f}$, and hence $\Omega_{f} \geq 0$ is an invariant set of the system.

With respect to the existence of limit sets, we first make a proposition about the late-time dynamics of the system:
\begin{prop}
Consider the dynamical system \eqref{eq:evolutionsys1} with parameters in the region $S1(F)$ defined by $-1 \leq w < \frac{1}{3}$, $\xi_{0} = 0$ and  $\eta_{0} = 0$.  Then, as $\tau \to +\infty$, $\Sigma^2 = \Sigma_{+}^2 + \Sigma_{-}^2 \to 0$ and $\mathcal{B}_{2}^2 \to 0$, and hence the model isotropizes.
\end{prop}
\begin{proof}
The details of the proof essentially follow the arguments given in the appendix of reference \cite{coleywainwright}. Substitution of Eq. \eqref{eq4:b1evolutionsys} in \eqref{eq:qdef2} results in the expression
\begin{equation}
\label{eq:qdef3}
q = \Sigma^2 \left( \frac{3-3w}{2} \right) + \mathcal{B}_{2}^2 \left(\frac{3 - 9w}{4}\right) + \frac{3w + 1}{2} - \frac{9}{2} \xi_{0}, 
\end{equation}
and hence the $\Omega_{f}'$ evolution equation \eqref{eq:OmegafP} may be written as
\begin{equation}
\label{eq:omegafp2}
\Omega_{f}' = \Omega_{f} \left[\Sigma^2 \left(3 - 3w\right) + \mathcal{B}_{2}^2 \left(\frac{3 - 9w}{2}\right) - 9 \xi_{0}    \right] + 4\eta_{0}\Sigma^2 + 9 \xi_{0}.
\end{equation}
In addition, from the generalized Friedmann equation, Eq. \eqref{eq4:b1evolutionsys}, we have that $\Omega_{f} \leq 1$.
Therefore, to prove the proposition it remains show that the function $\Omega_{f}$ is monotonically increasing, i.e., $\Omega_{f}' > 0$. 
Then,
\begin{equation}
\left[\Sigma^2 \left(3 - 3w\right) + \mathcal{B}_{2}^2 \left(\frac{3 - 9w}{2}\right)  \right] + 4\eta_{0}\Sigma^2 > 0 \Leftrightarrow -1 \leq w < \frac{1}{3}.
\end{equation}
Therefore, in the region where $\eta_{0} \geq 0$, $\xi_{0} = 0$, and $-1 \leq w < \frac{1}{3}$, $\Omega_{f}$ is monotonically increasing. In can therefore be said that in this region,
\begin{equation}
\lim_{\tau \to +\infty} \Omega_{f} = 1.
\end{equation}
Using this result in the Friedmann equation \eqref{eq4:b1evolutionsys}, we have that
\begin{equation}
\lim_{\tau \to +\infty} \Omega_{f} = 1 \Rightarrow \lim_{\tau \to +\infty} \left(-\frac{3}{2}\mathcal{B}_{2}^2 - \Sigma^{2}\right) = 0.
\end{equation}
The latter then implies that
\begin{equation}
\lim_{\tau \to +\infty} \Sigma^{2} = \lim_{\tau \to +\infty} \mathcal{B}_{2}^2 = 0.
\end{equation}
\end{proof}

In order to gain some insight into the asymptotic behavior of the system as $\tau \rightarrow -\infty$ we use the extended LaSalle principle for negatively invariant sets; see Proposition B.3. in reference \cite{hewwain}. In particular, suppose $\mathbf{x}' = \mathbf{f(x)}$ is an autonomous system of first-order differential equations and let $Z: \mathbb{R}^{n} \to \mathbb{R}$ be a $C^{1}$ function. If $S \subset \mathbb{R}^{n}$ is a closed, bounded, and negatively invariant set, and $Z'(\mathbf{x}) \equiv \nabla Z \cdot \mathbf{f(x)} \leq 0$, $\forall$ $\mathbf{x} \in S$, then the extended LaSalle principle states that $\forall$ $\mathbf{a} \in S$, $\alpha(\mathbf{a}) \subseteq \{ \mathbf{x} \in S | Z'(\mathbf{x}) = 0\}$. That is, the $\alpha$-limit set $\alpha(\mathbf{a})$ contains the local sources in the system at $\tau \to -\infty$. We use this principle to establish past asymptotic behavior in the following proposition.

\begin{prop}
For the dynamical system (\ref{eq:evolutionsys1}), $\alpha(\mathbf{a}) \subseteq \{\Omega_{f} = 0\} = \{\mathcal{K}\}$.
\end{prop}
\begin{proof}
The set $\{\Omega_{f} = 0\}$ is negatively invariant, closed, and bounded. From Eq. \eqref{eq:omegafp2} when $\Omega_{f} = 0$ it follows that $\Omega_{f}' \leq 0$ if and only if $\eta_{0} = \xi_{0} = 0$. Therefore, $\Omega_{f}' = 0$ if $\Omega_{f} = \xi_{0} = \eta_{0} = 0$, which is precisely the region defining the Kasner circle, so  $\alpha(\mathbf{a}) \subseteq \{\Omega_{f} = 0\}$.
\end{proof}


\subsection{Heteroclinic Orbits}
It is interesting to note that for the cosmological model under consideration in this paper, no finite heteroclinic sequences exist. The reason is that every heteroclinic sequence has an initial point that represents a local source, intermediate points which represent saddles, and a terminal point which represents a local sink. The caveat however, is that each equilibrium point and its corresponding asymptotic behavior must belong to the same region of the parameter space $(\eta_{0}, \xi_{0}, w)$, which is not possible for our dynamical system. There are however, several heteroclinic orbits which connect distinct equilibrium points, of which some have been plotted in Figs. (3), (4), and (5).
For the region defined by $\{(\eta_{0}, \xi_{0}, w) | (\eta_{0}, \xi_{0}, w) \in U(K) \cup S1(F) \}$, we have:
\begin{eqnarray}
	\mathcal{K} \rightarrow \mathcal{F}.
\end{eqnarray}
\begin{figure}[H]
\begin{center}
\label{fig:orbits1}
\caption{The heteroclinic orbits joining the $\mathcal{K} \rightarrow \mathcal{F}$. The plus signs indicate the Kasner equilibrium points that we found above, while the large circle indicates the FLRW equilibrium point. The numerical integration was completed with $\eta_{0} = \xi_{0} = 0$, $w = 0.325$.}
\includegraphics*[scale = 0.60]{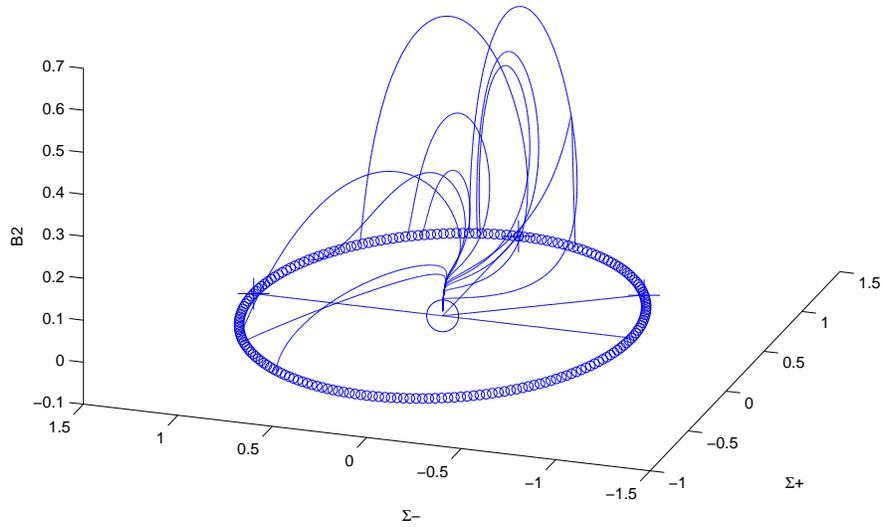}
\end{center}
\end{figure}
\newpage

For the region defined by $\{(\eta_{0}, \xi_{0}, w) | (\eta_{0}, \xi_{0}, w) \in U(K) \cup SA(F) \}$, we have:
\begin{eqnarray}
	\mathcal{K} \rightarrow \mathcal{F}.
\end{eqnarray}
\begin{figure}[H]
\begin{center}
\label{fig:orbits2}
\caption{The heteroclinic orbits joining the $\mathcal{K} \rightarrow \mathcal{F}$. The plus signs indicate the Kasner equilibrium points that we found above, while the large circle indicates the FLRW equilibrium point. The numerical integration was completed with $\eta_{0} = \xi_{0} = 0$, $w = \frac{1}{2}$.}
\includegraphics*[scale = 0.60]{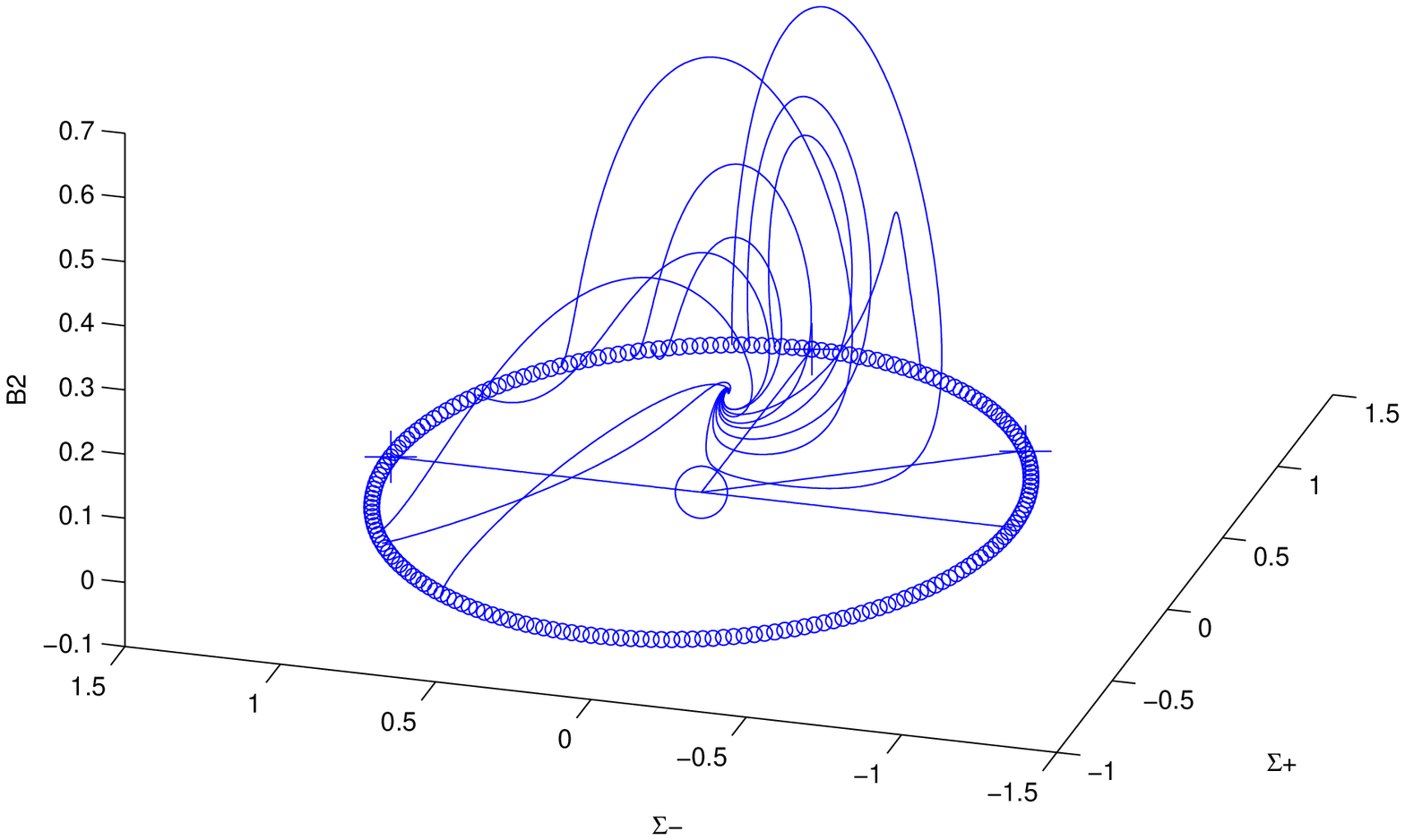}
\end{center}
\end{figure}
\newpage

For the region defined by $\{(\eta_{0}, \xi_{0}, w) | (\eta_{0}, \xi_{0}, w) \in SA(F) \cup S2(BIMV) \}$, we have
\begin{equation}
	\mathcal{F} \rightarrow \mathcal{BI_{MV}}.
\end{equation}
\begin{figure}[H]
\begin{center}
\label{fig:orbits3}
\caption{The heteroclinic orbits joining SA(F) to S2(BIMV). The circle represents the FLRW equilibrium point, while the star represents the $\mathcal{BI_{MV}}$ equilibrium point. The numerical integration was completed with $\eta_{0} = 2$,  $\xi_{0} = 0$, and $w = 0.40$.}
\includegraphics*[scale = 0.60]{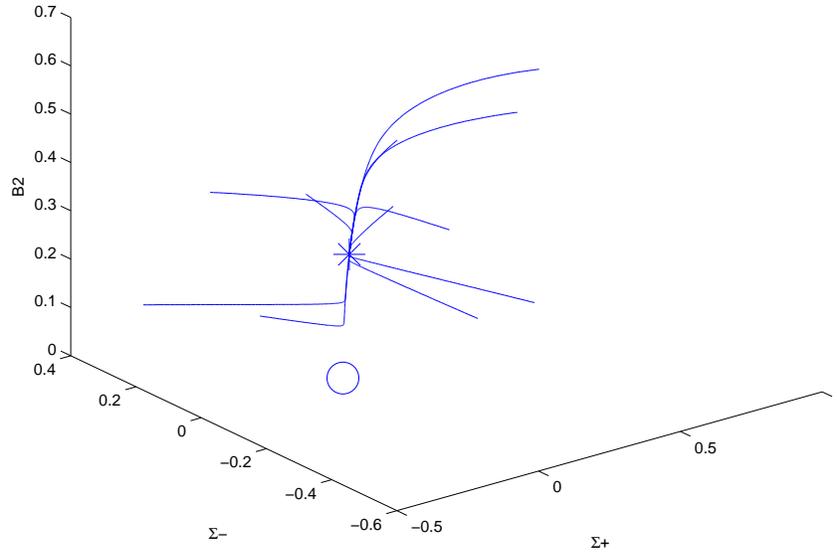}
\end{center}
\end{figure}
\newpage

\subsection{The General Case - Extending the Phase Space}
As discussed earlier, our work up to this point has assumed that the magnetic field is aligned along the shear eigenvector. The result of this approach was seen in Eq. \eqref{eq:Rdefs}, where to avoid $R_{1}, R_{2}$ or $R_{3}$ becoming singular we set $\mathcal{B}_{1} = \mathcal{B}_{3} = 0$, and $\mathcal{B}_{2} \neq 0$. Of course, this is not the most general case. 

For a Bianchi Type I universe with a magnetic field source, one can also consider the case for which the magnetic field is not a shear eigenvector as was done for the perfect fluid case by LeBlanc \cite{leblanc2}. The result of this approach is that the dynamical system is six-dimensional to accommodate additional non-diagonal shear components compared to just three dimensions with no non-diagonal shear components as is the case in our work. This extension of the phase space leads to dynamical equations that are indeed smooth over all phase space, with $R_{1}, R_{2}, R_{3}$ being continuous in general. 

With respect to qualitative behaviour, in LeBlanc's extended approach, he also obtains a flat FLRW equilibrium point, a new solution the Einstein field equations (via a previously undiscovered equilibrium point) and the Kasner vacuum (Page 2287, \cite{leblanc2}). He also concludes that a possible late-time future asymptotic state is a flat FLRW model (Page 2290, \cite{leblanc2}). Finally, LeBlanc also concludes that the Kasner circle is a past attractor (Page 2292, \cite{leblanc2}). Although LeBlanc obtains additional equilibrium points which is natural given the extension of the phase space dimension, the asymptotic qualitative behaviour he finds is the same as we have found in our work.

\section{A Numerical Analysis}

The goal of this section is to complement the preceding stability analysis of the equilibrium points with extensive numerical experiments in order to confirm that the local results are in fact global in nature. For each numerical simulation, we chose the initial conditions such that the constraint Eq. \eqref{eq4:b1evolutionsys} in addition to $\mathcal{B}_{2} \geq 0$ were satisfied.  Although numerical integrations were done from $0 \leq \tau \leq 3000$, for demonstration purposes we present solutions for shorter time intervals. We completed numerical integrations of the dynamical system for physically interesting cases of $w$ equal to $0$ (dust), $0.325$ (a dust/radiation mixture), and $1/3$ (radiation). Also note that in the subsequent plots, asterisks denote initial conditions. The actual initial conditions used can be found in the Appendix.
\newpage
\subsection{Dust Models: $w=0$}
\subsubsection{$\xi_{0} = 0.1, \eta_{0} = 0.2$}
\begin{figure}[H]
\begin{center}
\label{fig:fig2}
\caption{This figure shows the dynamical system behavior for $\xi_{0} = 0.1 $, $\eta_{0} = 0.2$, and $w = 0$. The diamond indicates the FLRW equilibrium point, and this numerical solution shows that it is a local sink of the dynamical system.  The model also isotropizes as can be seen from the last figure, where $\Sigma_{\pm}, \mathcal{B}_{2} \to 0$ as $\tau \to \infty$.}
\includegraphics*[scale = 0.60]{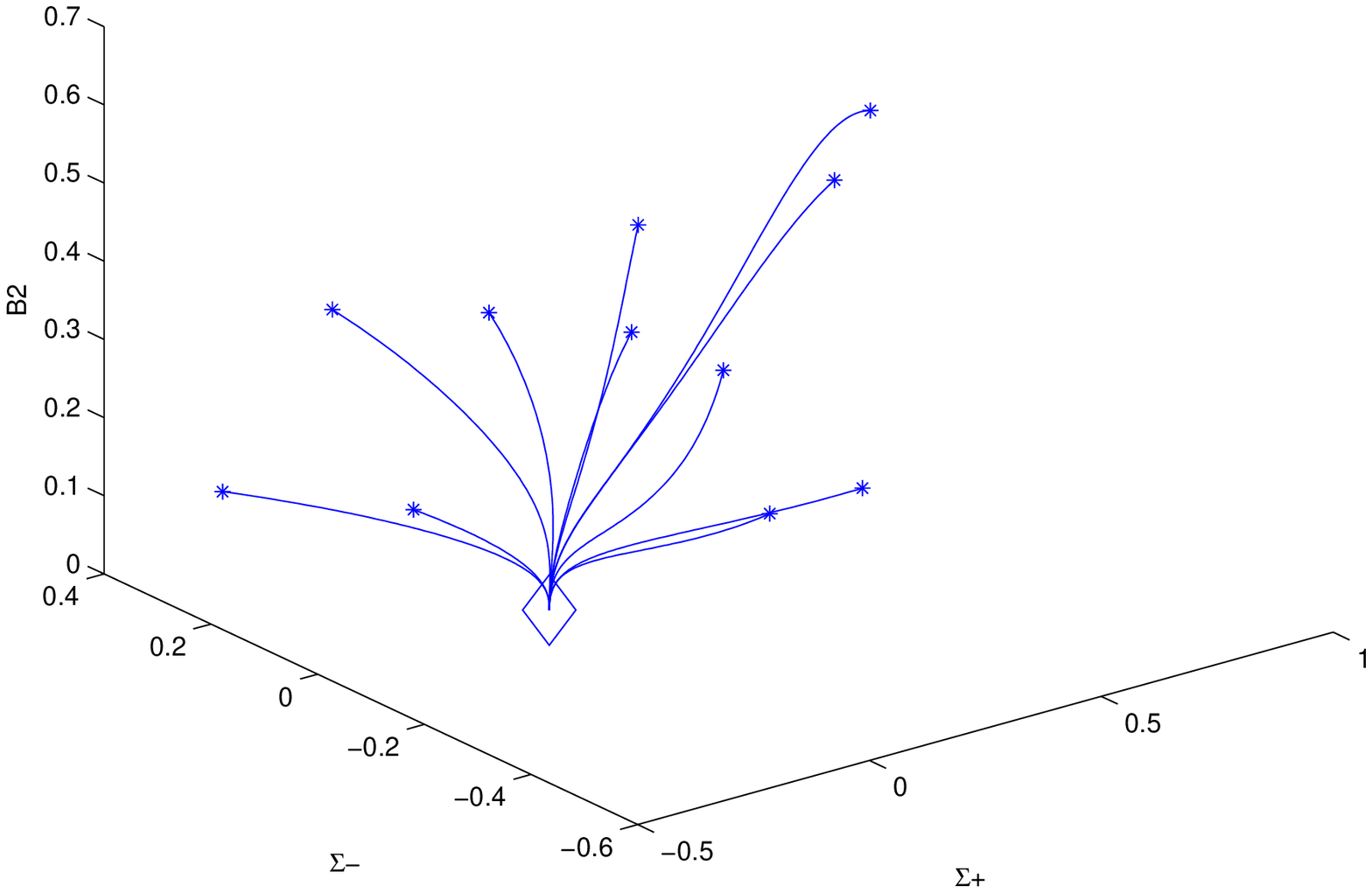} \\
\includegraphics*[scale = 0.60]{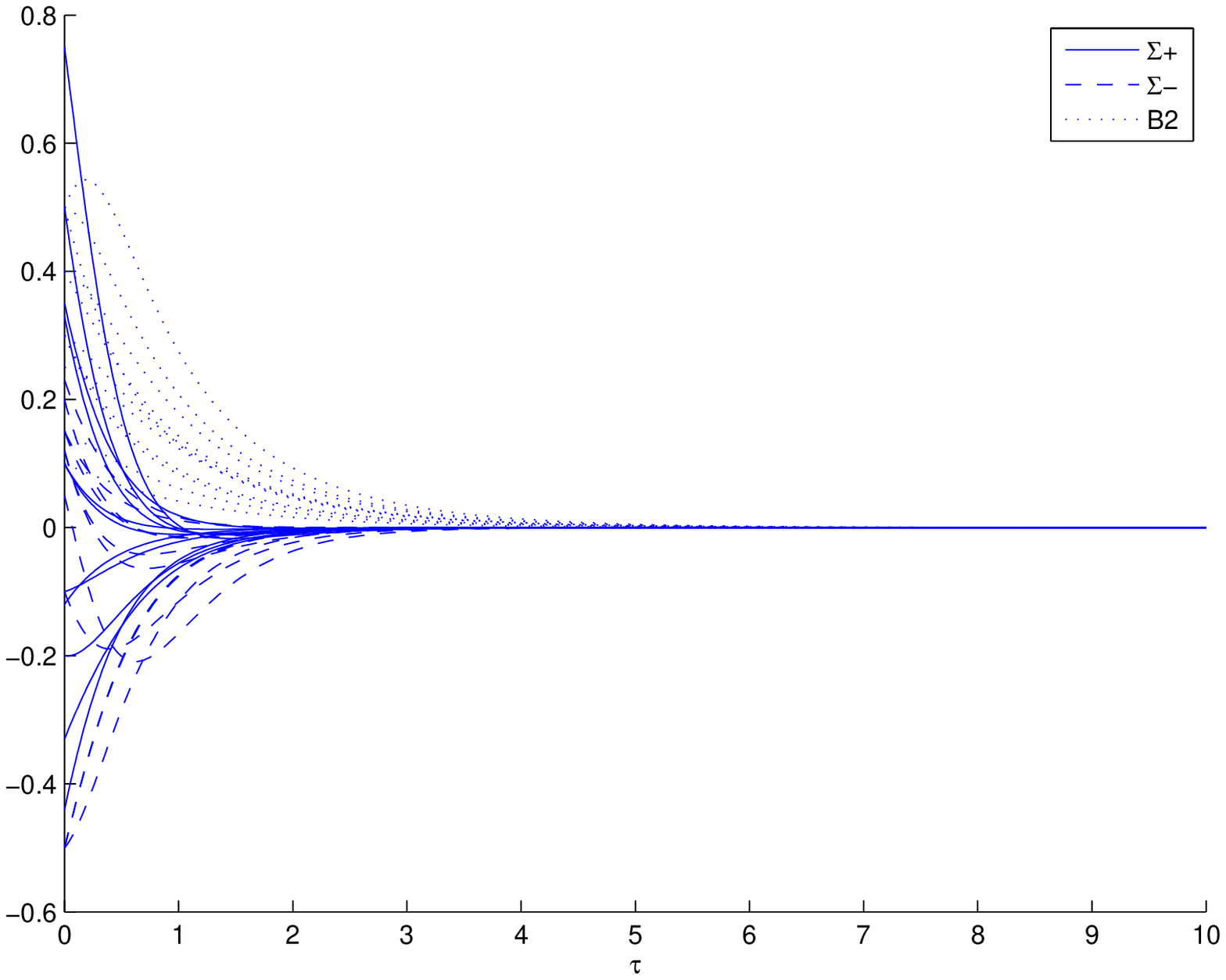}
\end{center}
\end{figure}
\newpage

\newpage
\subsubsection{$\xi_{0} = 1, \eta_{0} = 0.5$}
\begin{figure}[H]
\begin{center}
\label{fig:fig3}
\caption{This figure shows the dynamical system behavior for $\xi_{0} = 1 $, $\eta_{0} = 0.5$, and $w = 0$. The diamond indicates the FLRW equilibrium point, and this numerical solution shows that it is a local sink of the dynamical system.  The model also isotropizes as can be seen from the last figure, where $\Sigma_{\pm}, \mathcal{B}_{2} \to 0$ as $\tau \to \infty$.}
\includegraphics*[scale = 0.60]{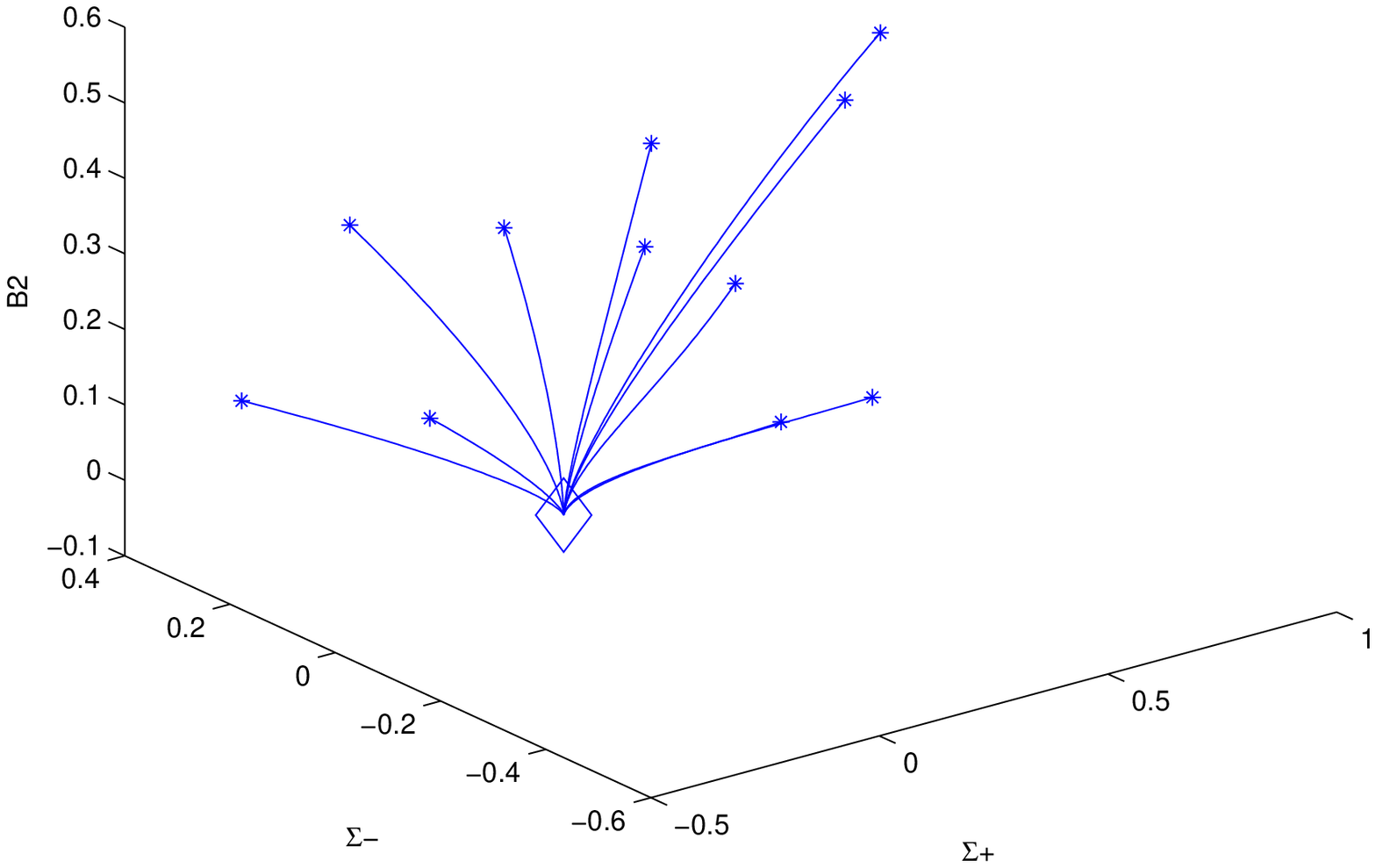} \\
\includegraphics*[scale = 0.60]{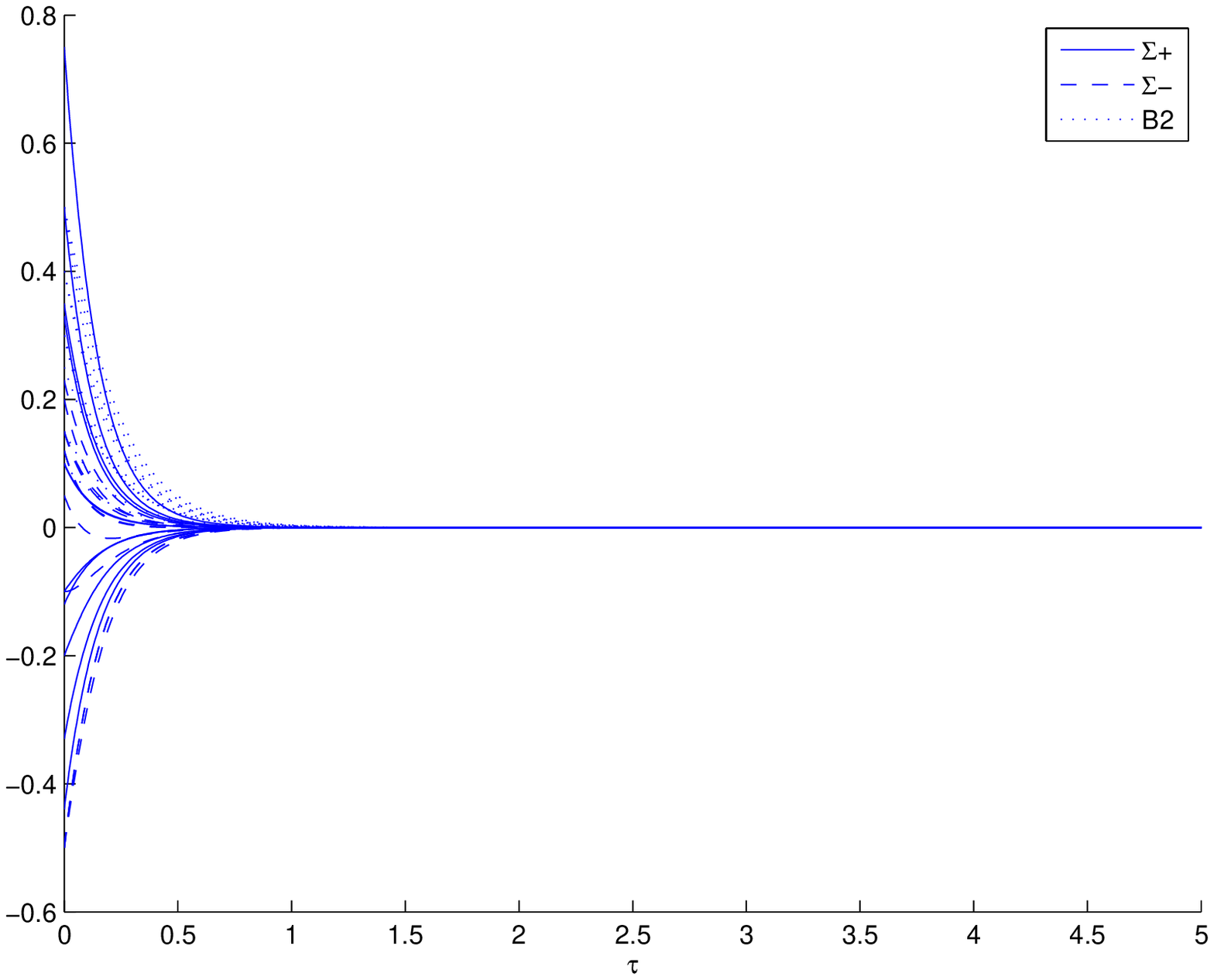}
\end{center}
\end{figure}

\newpage
\subsection{Radiation Models: $w = 1/3$}
\subsubsection{$\xi_{0} = 1.5, \eta_{0} =0$}
\begin{figure}[H]
\begin{center}
\label{fig:fig4}
\caption{This figure shows the dynamical system behavior for $\xi_{0} = 1.5$, $\eta_{0} = 0$, and $w = 1/3$. The diamond indicates the FLRW equilibrium point, and this numerical solution shows that it is a local sink of the dynamical system.  The model also isotropizes as can be seen from the last figure, where $\Sigma_{\pm}, \mathcal{B}_{2} \to 0$ as $\tau \to \infty$.}
\includegraphics*[scale = 0.60]{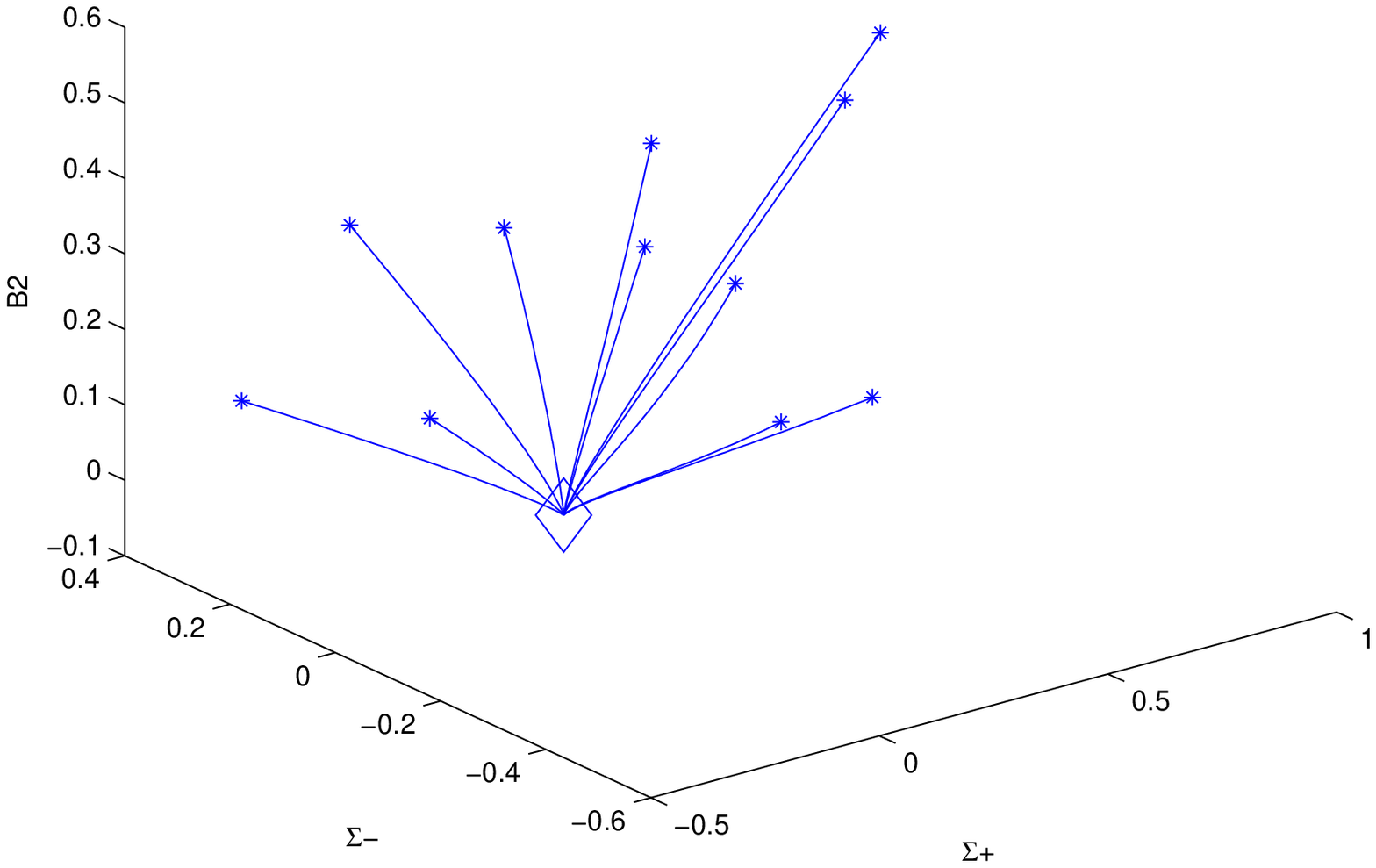} \\
\includegraphics*[scale = 0.60]{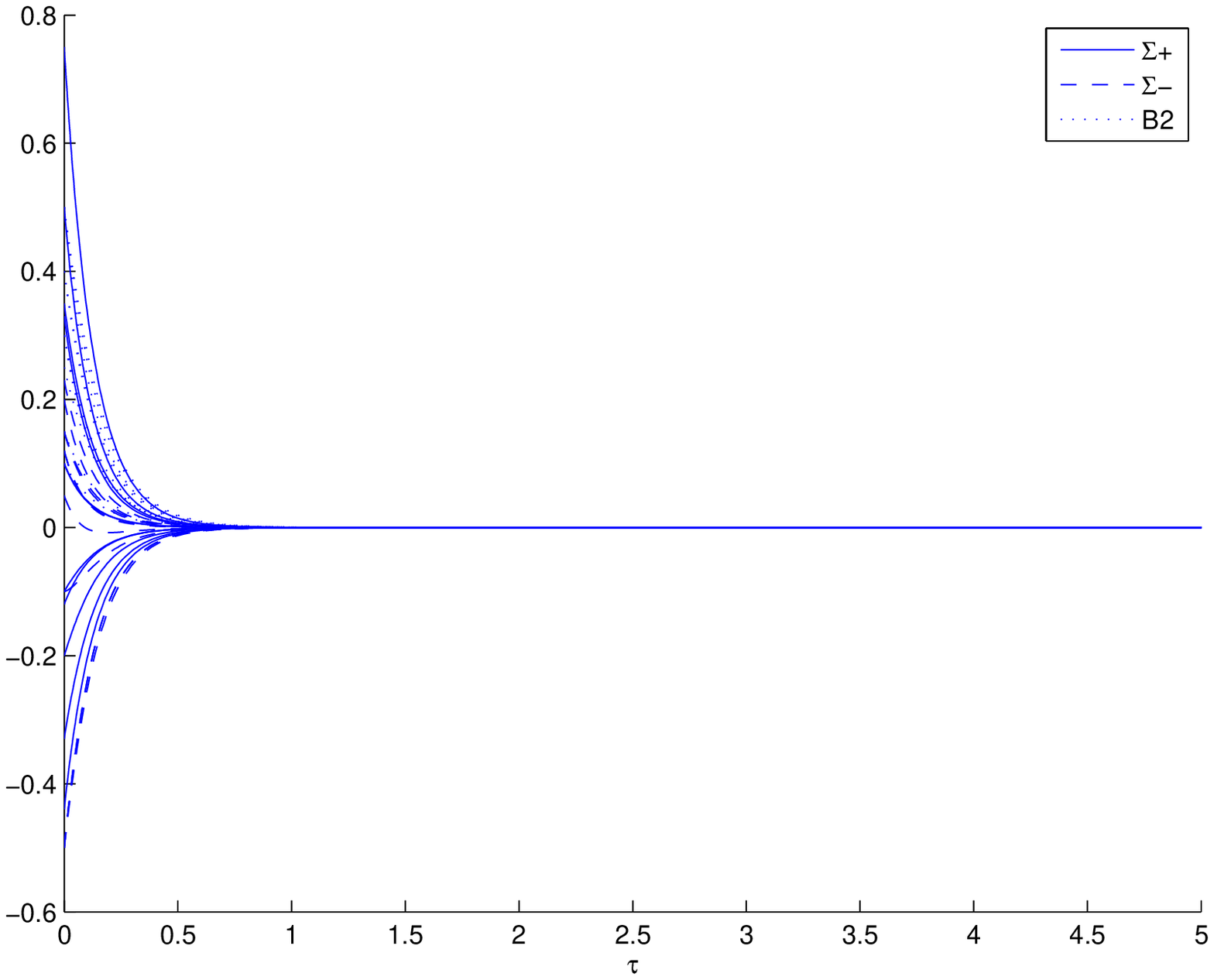}
\end{center}
\end{figure}

\newpage
\subsubsection{$\xi_{0} = 1.5, \eta_{0} = 0.5$}
\begin{figure}[H]
\begin{center}
\label{fig:fig5}
\caption{This figure shows the dynamical system behavior for $\xi_{0} = 1.5$, $\eta_{0} = 0.5$, and $w = 1/3$. The diamond indicates the FLRW equilibrium point, and this numerical solution shows that it is a local sink of the dynamical system.  The model also isotropizes as can be seen from the last figure, where $\Sigma_{\pm}, \mathcal{B}_{2} \to 0$ as $\tau \to \infty$.}
\includegraphics*[scale = 0.60]{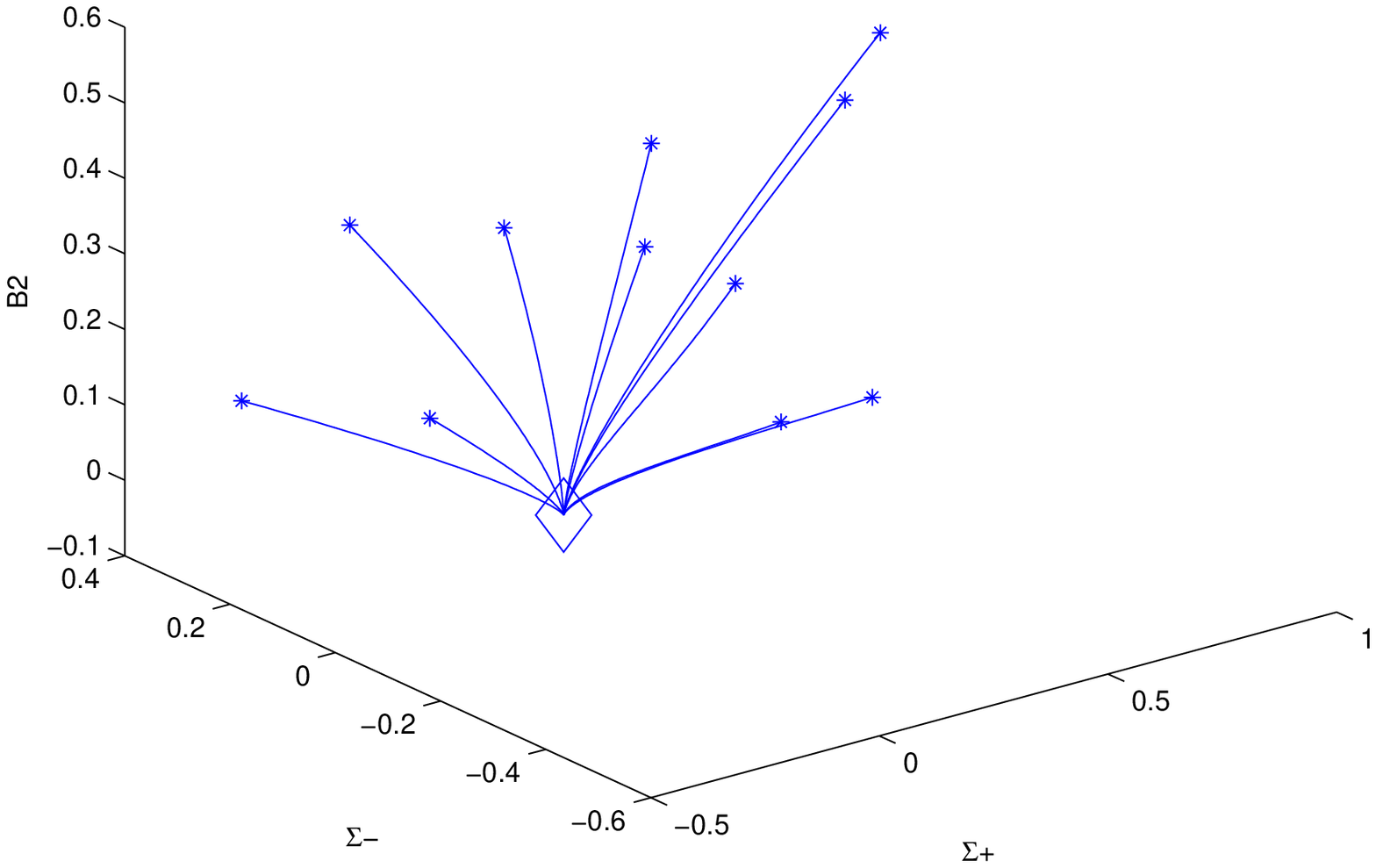} \\
\includegraphics*[scale = 0.60]{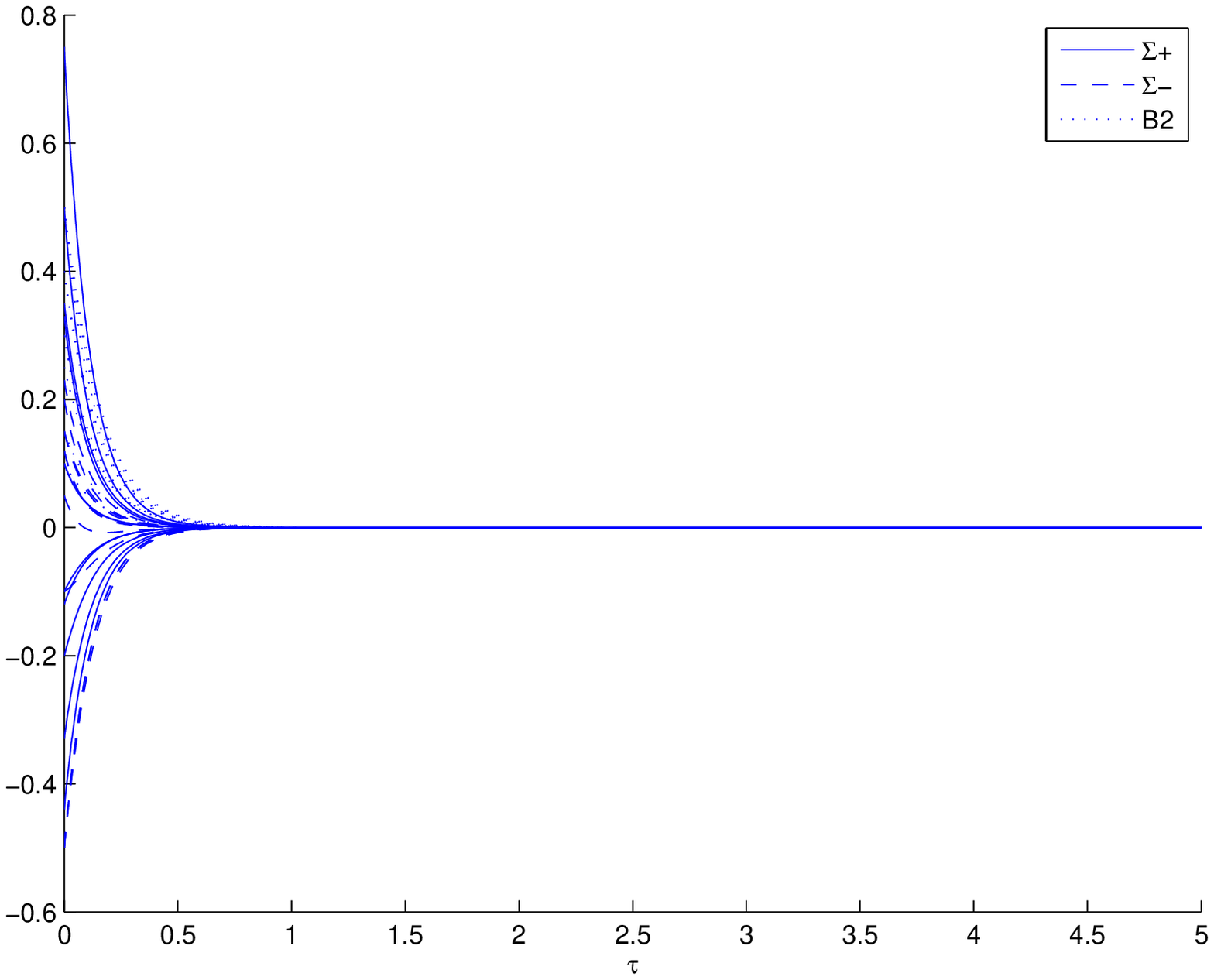}
\end{center}
\end{figure}

\newpage
\subsubsection{$\xi_{0} = 0, \eta_{0} = 2$}
\begin{figure}[H]
\begin{center}
\label{fig:fig6}
\caption{This figure shows the dynamical system behavior for $\xi_{0} = 0$, $\eta_{0} = 2$, and $w = 1/3$. The circle indicates the BIMV equilibrium point, and this numerical solution shows that it is a local sink of the dynamical system.  The model does not isotropize with respect to the anisotropic magnetic field as can be seen from the last figure, where $\mathcal_{B}_{2} > 0$ as $\tau \to \infty$, but does isotropize with respect to the spatial anisotropic variables, $\Sigma_{\pm},  \to 0$ as $\tau \to \infty$. This state is also special, since according to our fixed-point analysis, this behavior is only exhibited for $w = 1/3, \eta_{0} > 3/2$, and $\xi_{0} = 0$}. 
\includegraphics*[scale = 0.60]{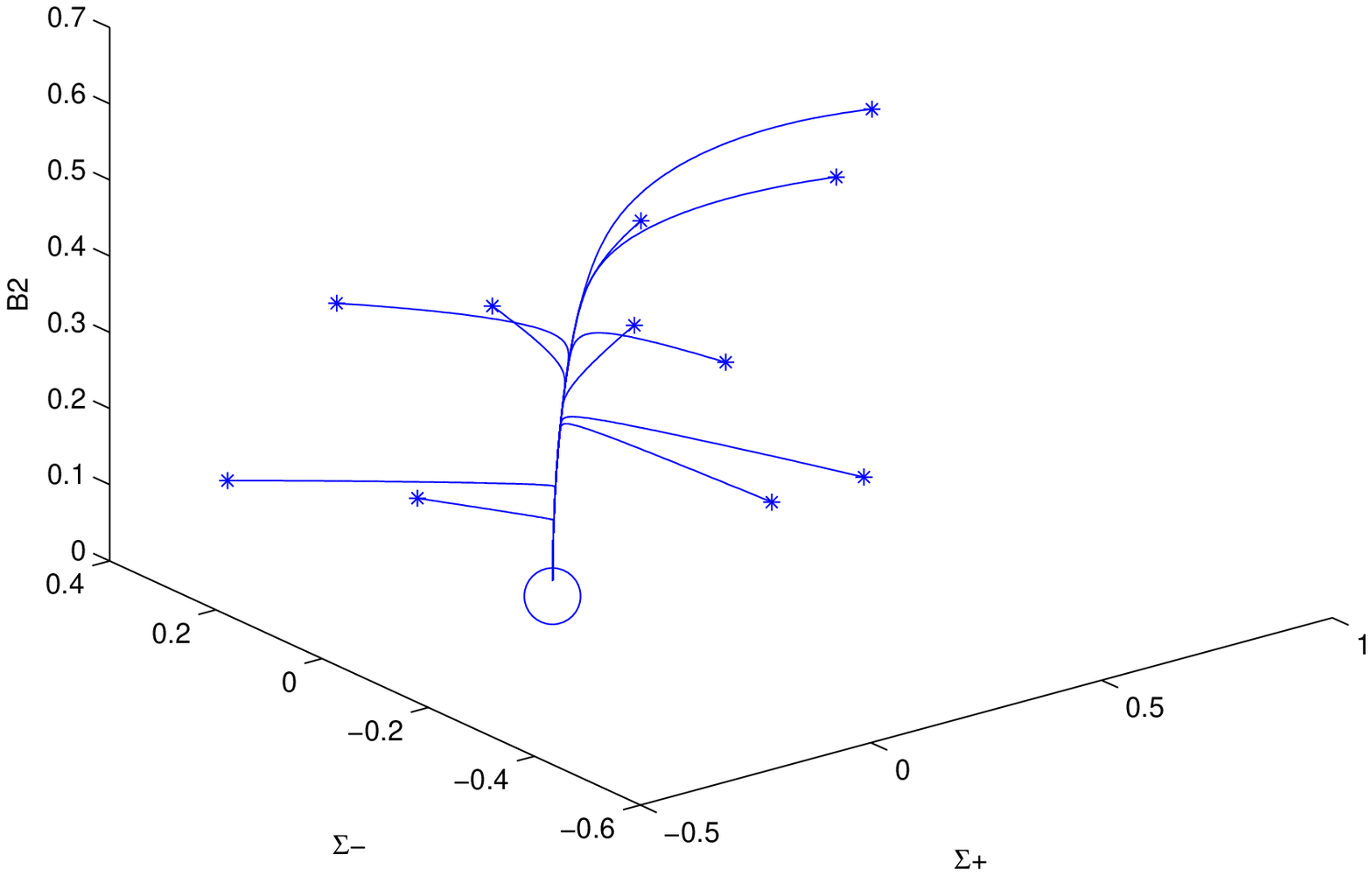} \\
\includegraphics*[scale = 0.60]{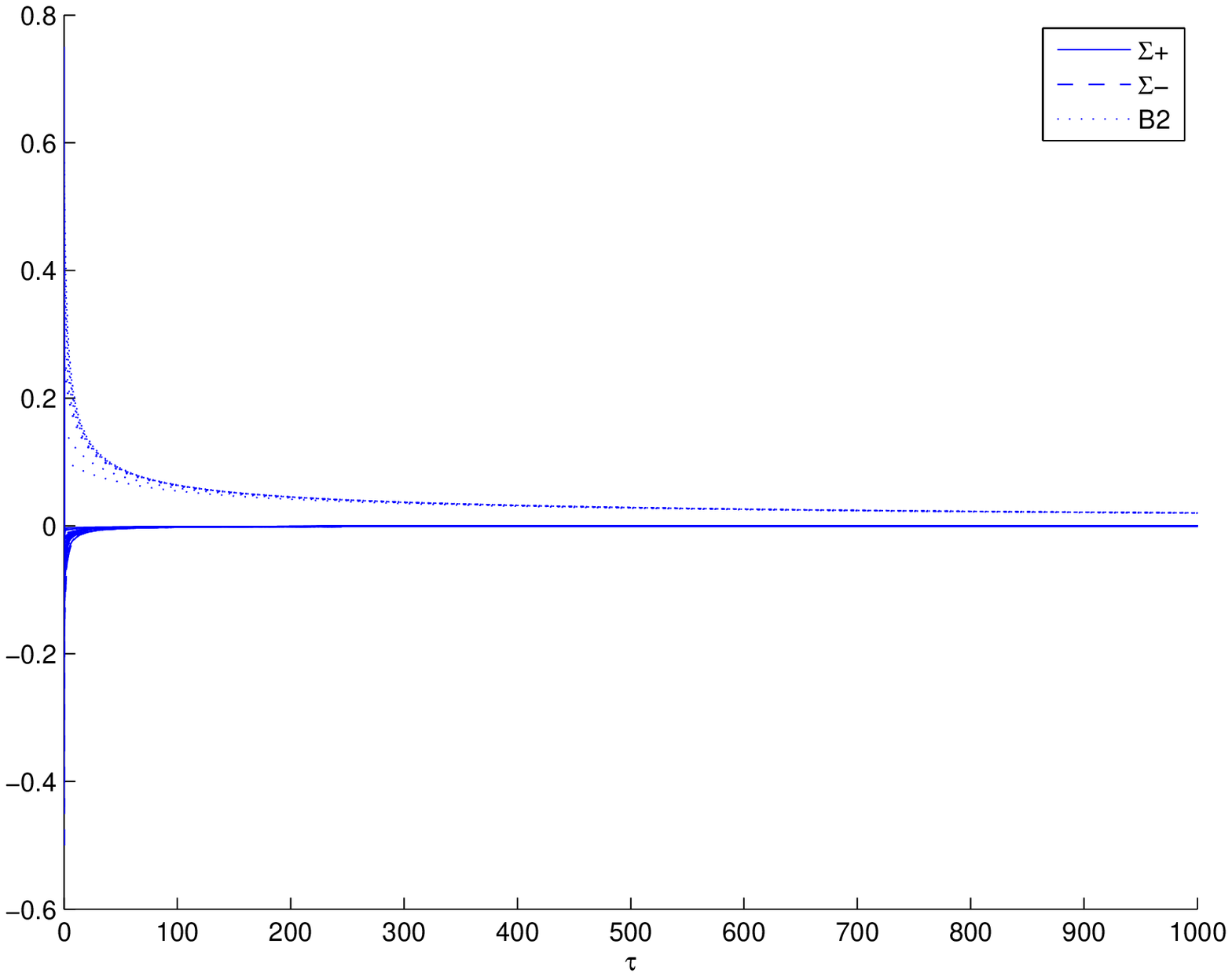}
\end{center}
\end{figure}

\newpage
\subsubsection{$\xi_{0} = 0, \eta_{0} = 10$}
\begin{figure}[H]
\begin{center}
\label{fig:fig7}
\caption{This figure shows the dynamical system behavior for $\xi_{0} = 0$, $\eta_{0} = 10$, and $w = 1/3$. The circle indicates the BIMV equilibrium point, and this numerical solution shows that it is a local sink of the dynamical system.  The model does not isotropize with respect to the anisotropic magnetic field as can be seen from the last figure, where $\mathcal_{B}_{2} > 0$ as $\tau \to \infty$, but does isotropize with respect to the spatial anisotropic variables, $\Sigma_{\pm},  \to 0$ as $\tau \to \infty$. This state is also special, since according to our fixed-point analysis, this behavior is only exhibited for $w = 1/3, \eta_{0} > 3/2$, and $\xi_{0} = 0$}.
\includegraphics*[scale = 0.60]{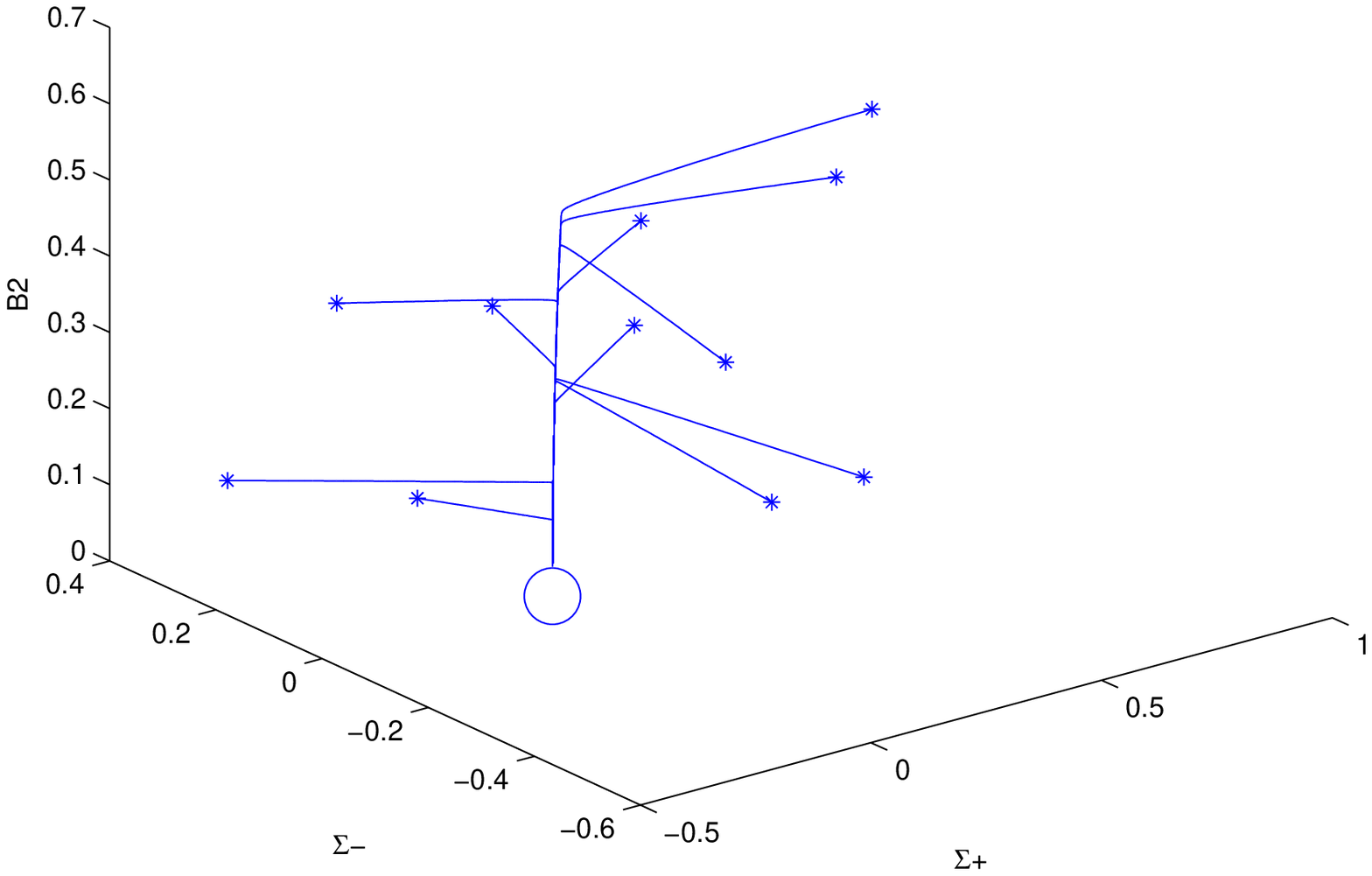} \\
\includegraphics*[scale = 0.60]{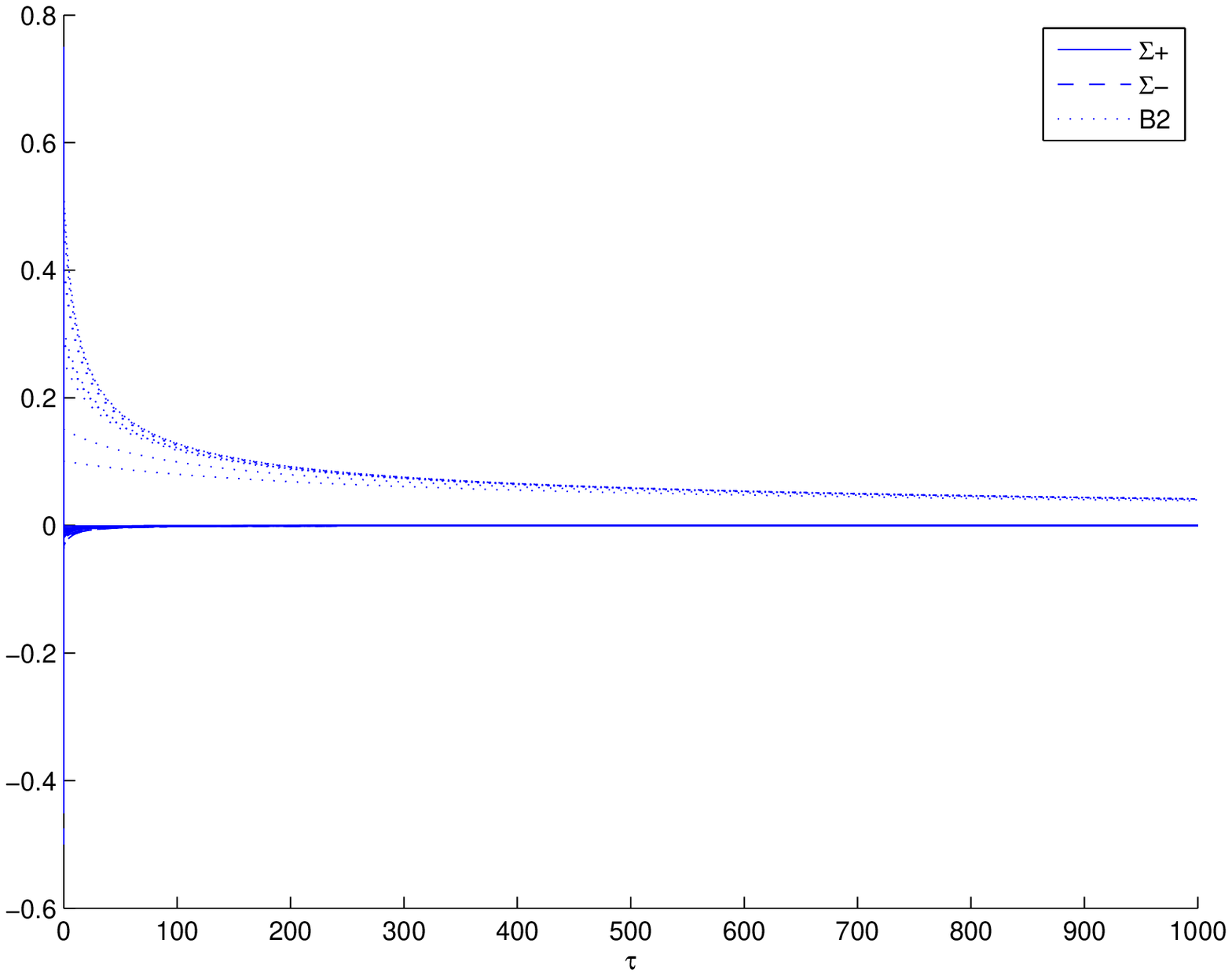}
\end{center}
\end{figure}

\newpage
\subsection{Dust/Radiation Models: $w = 0.325$}
\subsubsection{$\xi_{0} = 0.5, \eta_{0} = 0.5$}
\begin{figure}[H]
\begin{center}
\label{fig:fig8}
\caption{This figure shows the dynamical system behavior for $\xi_{0} = 0.5$, $\eta_{0} = 0.5$, and $w = 0.325$. The diamond indicates the FLRW equilibrium point, and this numerical solution shows that it is a local sink of the dynamical system.  The model also isotropizes as can be seen from the last figure, where $\Sigma_{\pm}, \mathcal{B}_{2} \to 0$ as $\tau \to \infty$.}. 
\includegraphics*[scale = 0.60]{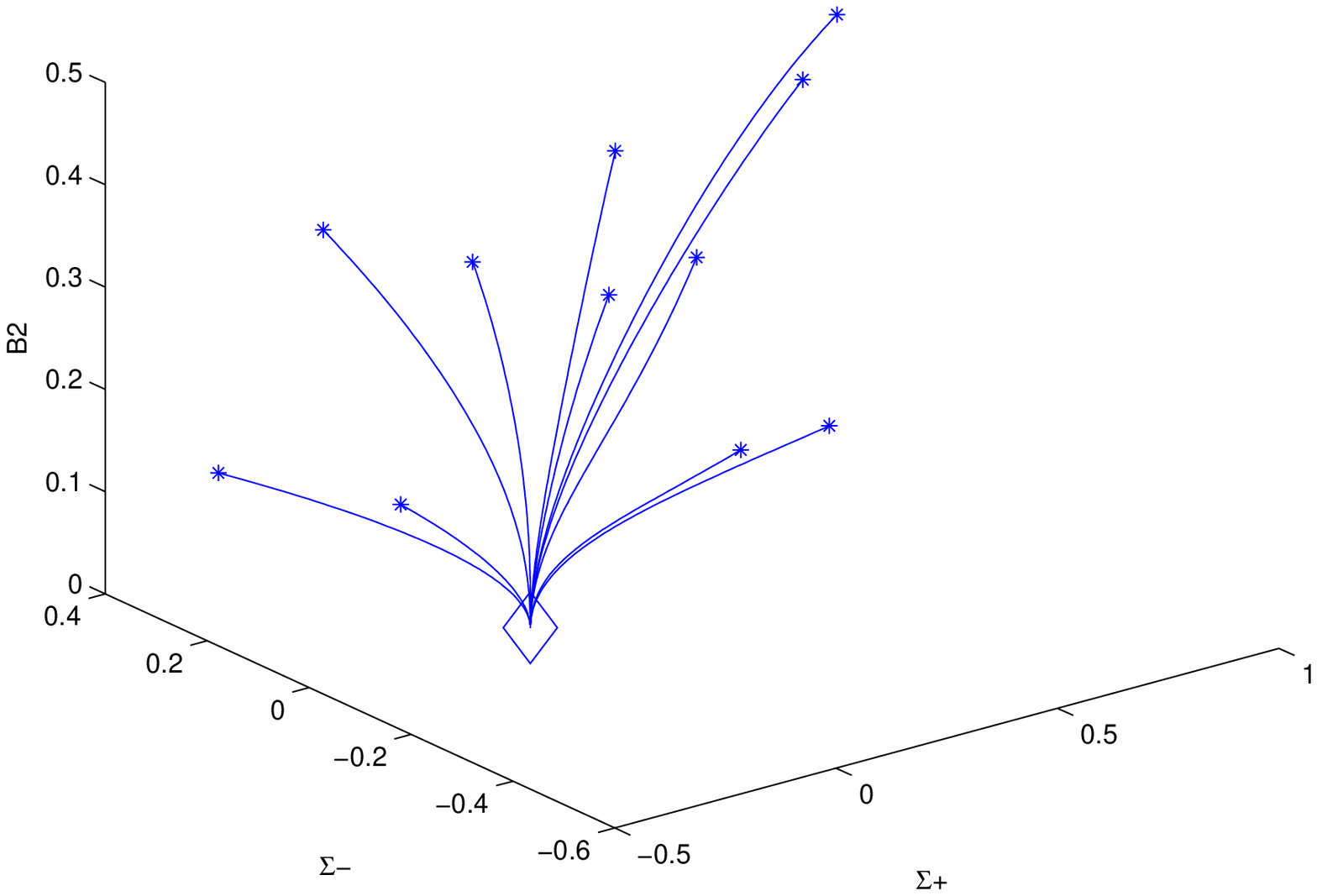} \\
\includegraphics*[scale = 0.60]{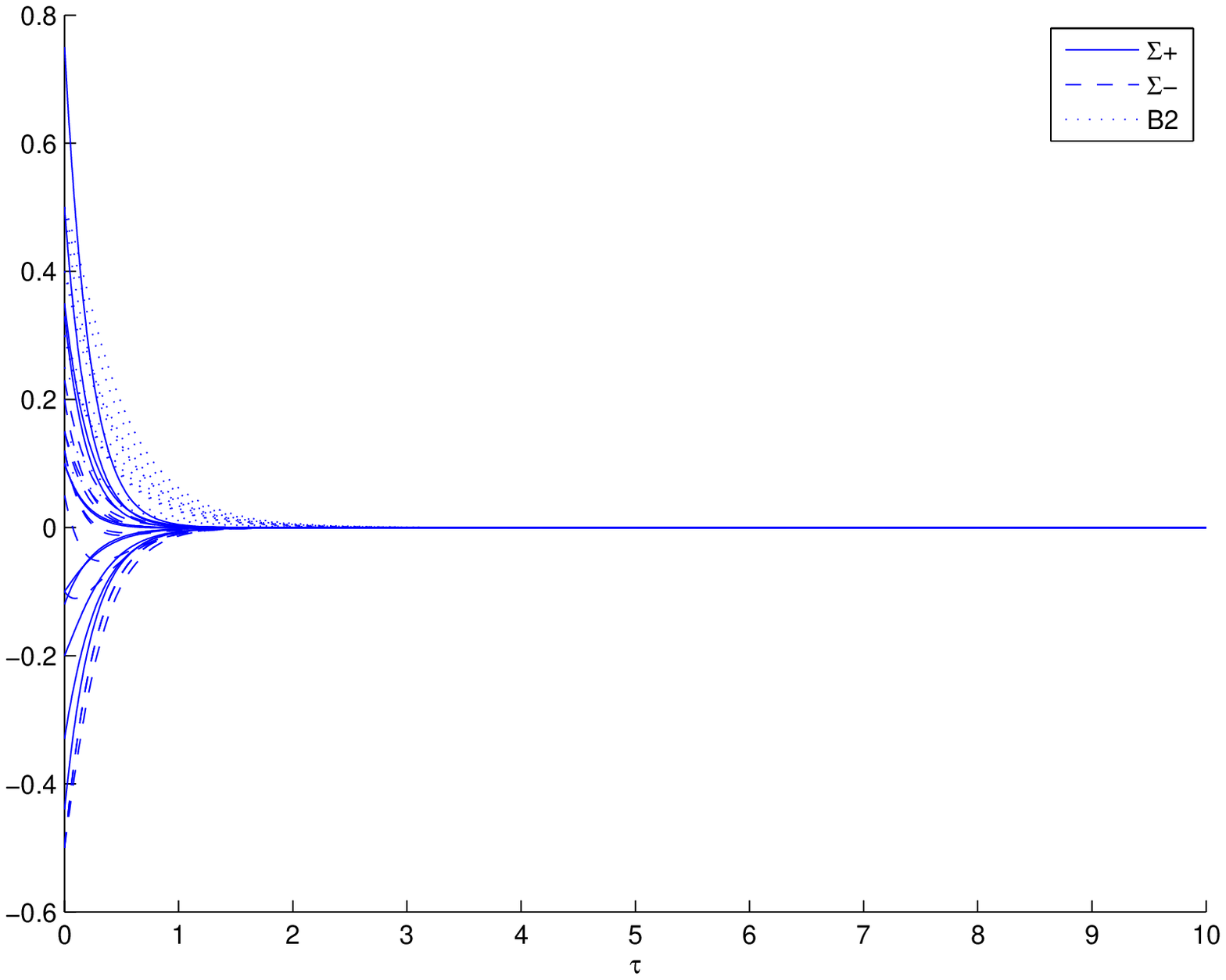}
\end{center}
\end{figure}

\newpage
\subsection{Dust/Radiation Models: $w = 0.325$}
\subsubsection{$\xi_{0} = 1, \eta_{0} = 2$}
\begin{figure}[H]
\begin{center}
\label{fig:fig9}
\caption{This figure shows the dynamical system behavior for $\xi_{0} = 1$, $\eta_{0} = 2$, and $w = 0.325$. The diamond indicates the FLRW equilibrium point, and this numerical solution shows that it is a local sink of the dynamical system.  The model also isotropizes as can be seen from the last figure, where $\Sigma_{\pm}, \mathcal{B}_{2} \to 0$ as $\tau \to \infty$.}. 
\includegraphics*[scale = 0.60]{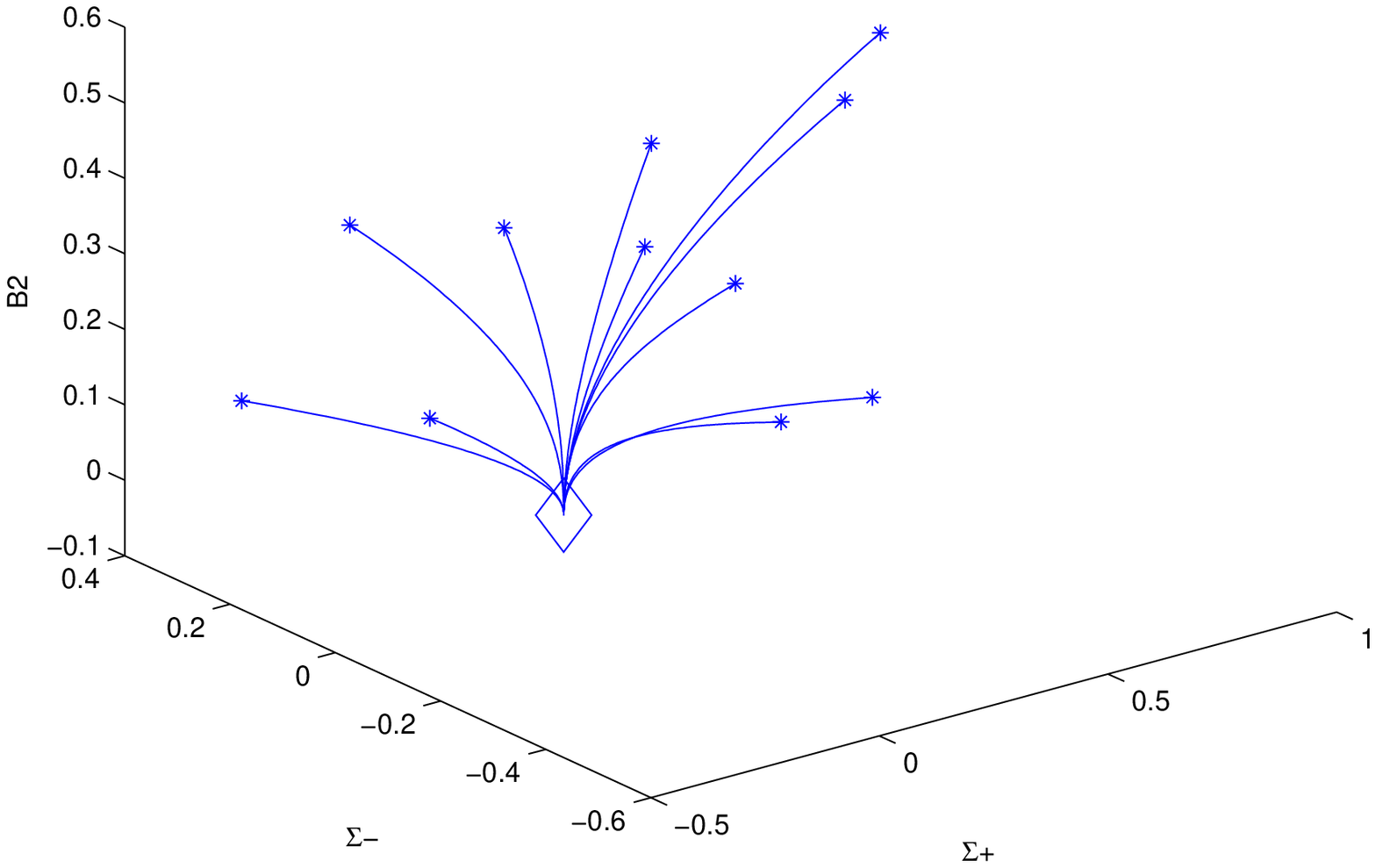} \\
\includegraphics*[scale = 0.60]{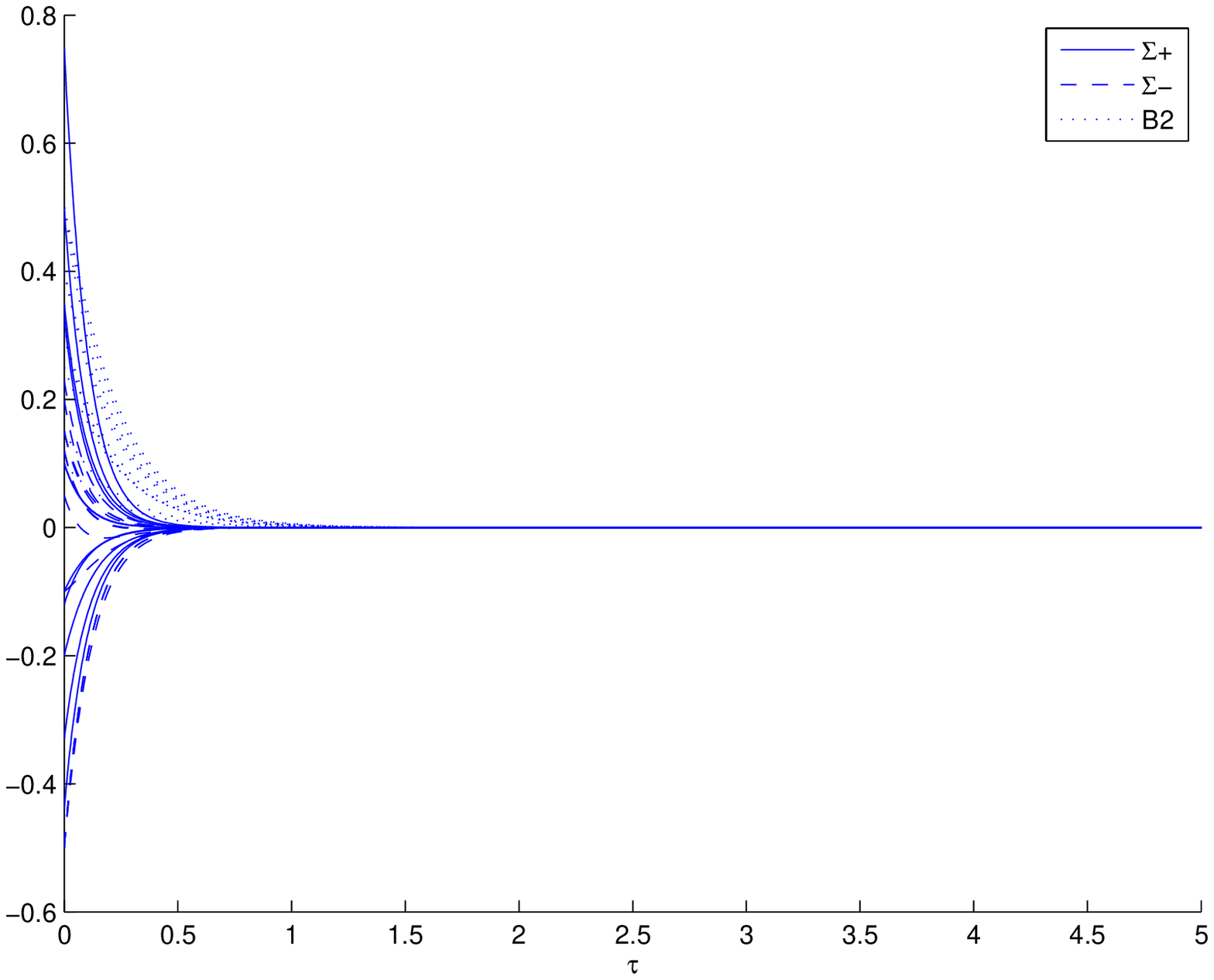}
\end{center}
\end{figure}


\newpage
\section{Conclusions}

We have presented in this paper a comprehensive analysis of the dynamical behavior of a Bianchi Type I viscous magnetohydrodynamic cosmology, using a variety of techniques ranging from a fixed point analysis to analyzing asymptotic behavior using standard dynamical systems theory combined with numerical experiments. We have shown that the fixed points may be associated with Kasner-type solutions, a flat universe FLRW solution, and interestingly, a new solution to the Einstein Field equations involving non-zero magnetic fields, and non-zero viscous coefficients.

For cases in which $\eta_{0} \geq 0, \quad \xi_0 \geq 0, \quad -1 \leq w < 1/3$ or $\eta_{0} \geq 0, \quad 1/3 \leq w \leq 1, \quad \xi_{0} > \left(3w-1\right)/9$, the dynamical model isotropizes asymptotically; that is, the spatial anisotropy and the anisotropic magnetic field decay to negligible values giving a close approximation to the present-day universe. We were also able to show that for regions in which $\eta_{0} > 3/2, \xi_{0} =0, w =1/3$ or $\eta_{0} > 3/2, \quad  1/3 < w < (6\eta_{0}-5)/(6\eta_{0}+3), \quad 0 \leq \xi_{0} \leq \left(3w-1\right)/9$, the model does not isotropize, rather at late times goes into a stable equilibrium in which there is a non-zero magnetic field. 

The flat FLRW model whose associated equilibrium point was denoted by $\mathcal{F}$, is of primary importance with respect to models of the present day universe. Through our fixed point analysis, we showed that $\mathcal{F}$ represents a saddle point if $\eta_{0} = 0, \quad 1/3 < w < 1, \quad 0 \leq \xi_{0} < (3w-1)/9$, $\eta_{0} = 0, \quad w = 1, \quad 0 < \xi_{0} < 2/9$, or $\eta_{0} > 0, \quad 1/3 < w \leq 1, \quad 0 \leq \xi_{0} < \left(3w-1\right)/9$, (which was denoted above by SA(F)). In these regions, $\mathcal{F}$ attracts along its stable manifold and repels along its unstable manifold. More precisely,  the stable manifold $W^{s}$ of the equilibrium point $\mathcal{F}$, is tangent to the stable subspace $E^{s}$ at $\mathcal{F}$ such that all orbits in $W^{s}$ approach $\mathcal{F}$ as $\tau \to \infty$. Similarly, there exists an unstable manifold $W^{u}$ of $\mathcal{F}$ such that it is tangent to the unstable subspace $E^{u}$ at $\mathcal{F}$ and such that all orbits in $W^{u}$ will approach $\mathcal{F}$ as $\tau \to -\infty$. Therefore, in the region denoted by SA(F), some orbits will have an initial attraction to $\mathcal{F}$, but will eventually be repelled by it. In the region denoted by S1(F), the point $\mathcal{F}$ is a local sink, and as such $\mathcal{F}$ attracts along its stable manifold, where the stable manifold $W^{s}$ of the equilibrium point $\mathcal{F}$, is tangent to the stable subspace $E^{s}$ at $\mathcal{F}$ such that all orbits in $W^{s}$ approach $\mathcal{F}$ as $\tau \to \infty$. There is therefore a time period, and two possible configurations for which the cosmological model will asymptotically isotropize, and be compatible with present-day observations of high-degree isotropy.


\section{Appendix}
\subsection{Jacobian Matrix for $\mathcal{BI_{MV}}$}

The Jacobian matrix for equilibrium point 3 is
\begin{equation}
\label{eq:Jacobian} J = \frac{1}{128\alpha}
\left[\begin{array}{ccc}
-(\alpha\mu_1+\mu_2\gamma) &  -\sqrt{3}\mu_3(\beta_1+\gamma)^2/(2\alpha) & 
\sqrt{3}(\beta_2-\gamma)^{1/2}(\alpha\mu_4+\mu_5\gamma)\\
-\sqrt{3}\mu_3(\beta_1+\gamma)^2/(2\alpha) &  -(\alpha\mu_6+\mu_7\gamma)  & 
3(\beta_2-\gamma)^{1/2}(\alpha\mu_4+\mu_5\gamma)\\
2\sqrt{3}(\beta_2-\gamma)^{1/2}(\alpha^2\mu_5+\mu_3\gamma)/3 &  
2(\beta_2-\gamma)^{1/2}(\alpha^2\mu_5+\mu_3\gamma) & 4\alpha\mu_5(\gamma-\beta_2)
\end{array}\right], 
\end{equation}
where, in addition to the definition of parameters in equations \eqref{eq:alpha}, \eqref{eq:beta1}, \eqref{eq:beta2} and \eqref{eq:gamma}, we define
\begin{equation}
\label{eq:mu_1_mu_2}
\mu_1 = 2\beta_1 -3(w-1)\beta_2 -144w(1+2\eta_0)-16(13-22\eta_0),
\quad
\mu_2 = 9 w^2 (1+2\eta_{0} ) + 12w(1-2\eta_{0}) -53+6\eta_0,
\end{equation}
\begin{equation}
\label{eq:mu_3_mu_4_mu_5}
\mu_3 = 3(w-1),\quad
\mu_4 = 9 w^2 (1+2\eta_{0} ) + 6w(3-2\eta_{0}) -39+2\eta_0,\quad
\mu_5 = 3w-1,
\end{equation}
\begin{equation}
\label{eq:mu_6_mu_7}
\mu_6 = 6\beta_1 -9(w-1)\beta_2 -240w(1+2\eta_0)-16(27-26\eta_0),\quad
\mu_7 = 27 w^2 (1+2\eta_{0} ) + 36w(1-2\eta_{0}) -95+18\eta_0.
\end{equation}
On the bifurcation surface $\xi_0 = (3w-1)/9$ we have the simplifications $\gamma = \beta_2 = -\beta_1$ and $\alpha\mu_1+\mu_2\gamma = \alpha\mu_6+\mu_7\gamma = 128\alpha(1+2\eta_0)$, and thus the matrix $J$ is diagonal.

\subsection{Initial Values for Numerical Experiments}
 \begin{table}[h]
\begin{center}\begin{tabular}{|c|c|c|c|}\hline 
$\Sigma_{+}$ & $\Sigma_{-}$ & $\mathcal{B}_{2}$ & $\Omega_{f}$
\\\hline 
0.1 & 0.2 & 0.3 & 0.8150 \\ \hline
0.1 & -0.5 & 0.3 & 0.6050 \\ \hline
-0.1 & -0.5 & 0.3 & 0.6050 \\ \hline
-0.2 & -0.5 & 0.5 & 0.3350 \\ \hline
0.5 & -0.1 & 0.5 & 0.3650 \\ \hline
0.75 & 0.05 & 0.5 & 0.0600 \\ \hline
0.33 & 0.12 & 0.4 & 0.6367 \\ \hline
-0.33 & 0.12 & 0.4 & 0.6367 \\ \hline
-0.44 & 0.32 & 0.15 & 0.7198 \\ \hline
-0.12 & 0.15 & 0.1 & 0.9481 \\ \hline
0.35 & 0.15 & 0.25 & 0.7613 \\ \hline
0.99 & 0 & 0 & 0.0199\\ \hline
0.499 & -0.855 & 0  & 0.0200 \\ \hline
0 & -0.99 & 0 & 0.0199 \\ \hline
0 & 0.99 & 0 & 0.0199 \\ \hline
\end{tabular} 
\caption{Initial conditions used in the numerical experiments. Note that in each case, $0 \leq \Omega_{f} \leq 1$ and $\mathcal{B}_{2} \geq 0$ as required.}
\end{center}
\label{defaulttable}
\end{table}

\newpage
\section{Acknowledgements}
The authors gratefully acknowledge the support of the Natural Sciences and Research Council of Canada. We would also like to thank the referee for his/her helpful suggestions upon reviewing this paper.

\newpage 
\bibliography{sources}

\begin{thebibliography}{34}
\expandafter\ifx\csname natexlab\endcsname\relax\def\natexlab#1{#1}\fi
\expandafter\ifx\csname bibnamefont\endcsname\relax
  \def\bibnamefont#1{#1}\fi
\expandafter\ifx\csname bibfnamefont\endcsname\relax
  \def\bibfnamefont#1{#1}\fi
\expandafter\ifx\csname citenamefont\endcsname\relax
  \def\citenamefont#1{#1}\fi
\expandafter\ifx\csname url\endcsname\relax
  \def\url#1{\texttt{#1}}\fi
\expandafter\ifx\csname urlprefix\endcsname\relax\def\urlprefix{URL }\fi
\providecommand{\bibinfo}[2]{#2}
\providecommand{\eprint}[2][]{\url{#2}}

\bibitem[{\citenamefont{Wainwright and Ellis}(1997)}]{ellis}
\bibinfo{author}{\bibfnamefont{J.}~\bibnamefont{Wainwright}} \bibnamefont{and}
  \bibinfo{author}{\bibfnamefont{G.}~\bibnamefont{Ellis}},
  \emph{\bibinfo{title}{Dynamical Systems in Cosmology}}
  (\bibinfo{publisher}{Cambridge University Press}, \bibinfo{year}{1997}),
  \bibinfo{edition}{1st} ed.

\bibitem[{\citenamefont{Gr{\o}n and Hervik}(2007)}]{hervik}
\bibinfo{author}{\bibfnamefont{{\O}.}~\bibnamefont{Gr{\o}n}} \bibnamefont{and}
  \bibinfo{author}{\bibfnamefont{S.}~\bibnamefont{Hervik}},
  \emph{\bibinfo{title}{Einstein's General Theory of Relativity: With Modern
  Applications in Cosmology}} (\bibinfo{publisher}{Springer},
  \bibinfo{year}{2007}), \bibinfo{edition}{1st} ed.

\bibitem[{\citenamefont{Grasso and Rubinstein}(2001)}]{grassorub}
\bibinfo{author}{\bibfnamefont{D.}~\bibnamefont{Grasso}} \bibnamefont{and}
  \bibinfo{author}{\bibfnamefont{H.~R.} \bibnamefont{Rubinstein}},
  \bibinfo{journal}{Physics Reports} \textbf{\bibinfo{volume}{348}},
  \bibinfo{pages}{163} (\bibinfo{year}{2001}).

\bibitem[{\citenamefont{Ando and Kusenko}(2010)}]{andokusenko}
\bibinfo{author}{\bibfnamefont{S.}~\bibnamefont{Ando}} \bibnamefont{and}
  \bibinfo{author}{\bibfnamefont{A.}~\bibnamefont{Kusenko}},
  \bibinfo{journal}{The Astrophysical Journal Letters}
  \textbf{\bibinfo{volume}{722}}, \bibinfo{pages}{L39} (\bibinfo{year}{2010}).

\bibitem[{\citenamefont{Gregori and et~al.}(2012)}]{gregorietal}
\bibinfo{author}{\bibfnamefont{G.}~\bibnamefont{Gregori}} \bibnamefont{and}
  \bibinfo{author}{\bibnamefont{et~al.}}, \bibinfo{journal}{Nature}
  \textbf{\bibinfo{volume}{481}}, \bibinfo{pages}{480} (\bibinfo{year}{2012}).

\bibitem[{\citenamefont{Schlickeiser}(2012)}]{Schlickeiser}
\bibinfo{author}{\bibfnamefont{R.}~\bibnamefont{Schlickeiser}},
  \bibinfo{journal}{Physical Review Letters} \textbf{\bibinfo{volume}{109}},
  \bibinfo{pages}{261101} (\bibinfo{year}{2012}).

\bibitem[{\citenamefont{Ellis et~al.}(2012)\citenamefont{Ellis, Maartens, and
  MacCallum}}]{elliscosmo}
\bibinfo{author}{\bibfnamefont{G.~F.} \bibnamefont{Ellis}},
  \bibinfo{author}{\bibfnamefont{R.}~\bibnamefont{Maartens}}, \bibnamefont{and}
  \bibinfo{author}{\bibfnamefont{M.~A.} \bibnamefont{MacCallum}},
  \emph{\bibinfo{title}{Relativistic Cosmology}} (\bibinfo{publisher}{Cambridge
  University Press}, \bibinfo{year}{2012}), \bibinfo{edition}{1st} ed.

\bibitem[{\citenamefont{Hughston and Jacobs}(1970)}]{hughstonjacobs}
\bibinfo{author}{\bibfnamefont{L.~P.} \bibnamefont{Hughston}} \bibnamefont{and}
  \bibinfo{author}{\bibfnamefont{K.~C.} \bibnamefont{Jacobs}},
  \bibinfo{journal}{Astrophysical Journal} \textbf{\bibinfo{volume}{160}},
  \bibinfo{pages}{147} (\bibinfo{year}{1970}).

\bibitem[{\citenamefont{LeBlanc}(1998)}]{leblanc1}
\bibinfo{author}{\bibfnamefont{V.}~\bibnamefont{LeBlanc}},
  \bibinfo{journal}{Classical and Quantum Gravity}
  \textbf{\bibinfo{volume}{15}}, \bibinfo{pages}{1607} (\bibinfo{year}{1998}).

\bibitem[{\citenamefont{LeBlanc}(1997)}]{leblanc2}
\bibinfo{author}{\bibfnamefont{V.}~\bibnamefont{LeBlanc}},
  \bibinfo{journal}{Classical and Quantum Gravity}
  \textbf{\bibinfo{volume}{14}}, \bibinfo{pages}{2281} (\bibinfo{year}{1997}).

\bibitem[{\citenamefont{Collins}(1972)}]{collins}
\bibinfo{author}{\bibfnamefont{C.}~\bibnamefont{Collins}},
  \bibinfo{journal}{Communications in Mathematical Physics}
  \textbf{\bibinfo{volume}{27}}, \bibinfo{pages}{37} (\bibinfo{year}{1972}).

\bibitem[{\citenamefont{LeBlanc et~al.}(1995)\citenamefont{LeBlanc, Kerr, and
  Wainwright}}]{leblanc3}
\bibinfo{author}{\bibfnamefont{V.}~\bibnamefont{LeBlanc}},
  \bibinfo{author}{\bibfnamefont{D.}~\bibnamefont{Kerr}}, \bibnamefont{and}
  \bibinfo{author}{\bibfnamefont{J.}~\bibnamefont{Wainwright}},
  \bibinfo{journal}{Classical and Quantum Gravity}
  \textbf{\bibinfo{volume}{12}}, \bibinfo{pages}{513} (\bibinfo{year}{1995}).

\bibitem[{\citenamefont{Barrow et~al.}(2007)\citenamefont{Barrow, Maartens, and
  Tsagas}}]{barrowmaartenstsagas}
\bibinfo{author}{\bibfnamefont{J.~D.} \bibnamefont{Barrow}},
  \bibinfo{author}{\bibfnamefont{R.}~\bibnamefont{Maartens}}, \bibnamefont{and}
  \bibinfo{author}{\bibfnamefont{C.~G.} \bibnamefont{Tsagas}},
  \bibinfo{journal}{Physics Reports} \textbf{\bibinfo{volume}{449}},
  \bibinfo{pages}{131} (\bibinfo{year}{2007}).

\bibitem[{\citenamefont{van Leeuwen and Salvati}(1985)}]{vanLeeuwen1}
\bibinfo{author}{\bibfnamefont{W.}~\bibnamefont{van Leeuwen}} \bibnamefont{and}
  \bibinfo{author}{\bibfnamefont{G.}~\bibnamefont{Salvati}},
  \bibinfo{journal}{Annals of Physics} \textbf{\bibinfo{volume}{165}},
  \bibinfo{pages}{214} (\bibinfo{year}{1985}).

\bibitem[{\citenamefont{Banerjee and Sanyal}(1986)}]{banerjeesanyal}
\bibinfo{author}{\bibfnamefont{A.}~\bibnamefont{Banerjee}} \bibnamefont{and}
  \bibinfo{author}{\bibfnamefont{A.}~\bibnamefont{Sanyal}},
  \bibinfo{journal}{General Relativity and Gravitation}
  \textbf{\bibinfo{volume}{18}}, \bibinfo{pages}{1251} (\bibinfo{year}{1986}).

\bibitem[{\citenamefont{Benton and Tupper}(1986)}]{bentontupper}
\bibinfo{author}{\bibfnamefont{J.}~\bibnamefont{Benton}} \bibnamefont{and}
  \bibinfo{author}{\bibfnamefont{B.}~\bibnamefont{Tupper}},
  \bibinfo{journal}{Physical Review D} \textbf{\bibinfo{volume}{18}},
  \bibinfo{pages}{1251} (\bibinfo{year}{1986}).

\bibitem[{\citenamefont{Salvati et~al.}(1987)\citenamefont{Salvati, Schelling,
  and van Leeuwen}}]{vanLeeuwen2}
\bibinfo{author}{\bibfnamefont{G.}~\bibnamefont{Salvati}},
  \bibinfo{author}{\bibfnamefont{E.}~\bibnamefont{Schelling}},
  \bibnamefont{and} \bibinfo{author}{\bibfnamefont{W.}~\bibnamefont{van
  Leeuwen}}, \bibinfo{journal}{Annals of Physics}
  \textbf{\bibinfo{volume}{179}}, \bibinfo{pages}{52} (\bibinfo{year}{1987}).

\bibitem[{\citenamefont{Sanyal and Ribeiro}(1987)}]{ribeirosanyal}
\bibinfo{author}{\bibfnamefont{A.}~\bibnamefont{Sanyal}} \bibnamefont{and}
  \bibinfo{author}{\bibfnamefont{M.}~\bibnamefont{Ribeiro}},
  \bibinfo{journal}{Journal of Mathematical Physics}
  \textbf{\bibinfo{volume}{28}}, \bibinfo{pages}{657} (\bibinfo{year}{1987}).

\bibitem[{\citenamefont{van Leeuwenn et~al.}(1989)\citenamefont{van Leeuwenn,
  Miedema, and Wiersma}}]{vanLeeuwen4}
\bibinfo{author}{\bibfnamefont{W.}~\bibnamefont{van Leeuwenn}},
  \bibinfo{author}{\bibfnamefont{P.}~\bibnamefont{Miedema}}, \bibnamefont{and}
  \bibinfo{author}{\bibfnamefont{S.}~\bibnamefont{Wiersma}},
  \bibinfo{journal}{General Relativity and Gravitation}
  \textbf{\bibinfo{volume}{21}}, \bibinfo{pages}{413} (\bibinfo{year}{1989}).

\bibitem[{\citenamefont{Pradhan and Pandey}(2003)}]{pradhanpandey}
\bibinfo{author}{\bibfnamefont{A.}~\bibnamefont{Pradhan}} \bibnamefont{and}
  \bibinfo{author}{\bibfnamefont{O.}~\bibnamefont{Pandey}},
  \bibinfo{journal}{International Journal of Modern Physics D}
  \textbf{\bibinfo{volume}{12}}, \bibinfo{pages}{1299} (\bibinfo{year}{2003}).

\bibitem[{\citenamefont{Pradhan and Singh}(2004)}]{pradhansingh}
\bibinfo{author}{\bibfnamefont{A.}~\bibnamefont{Pradhan}} \bibnamefont{and}
  \bibinfo{author}{\bibfnamefont{S.}~\bibnamefont{Singh}},
  \bibinfo{journal}{International Journal of Modern Physics D}
  \textbf{\bibinfo{volume}{13}}, \bibinfo{pages}{503} (\bibinfo{year}{2004}).

\bibitem[{\citenamefont{Bali and Anjali}(2004)}]{balianjali}
\bibinfo{author}{\bibfnamefont{R.}~\bibnamefont{Bali}} \bibnamefont{and}
  \bibinfo{author}{\bibnamefont{Anjali}}, \bibinfo{journal}{Pramana-Journal of
  Physics} \textbf{\bibinfo{volume}{63}}, \bibinfo{pages}{481}
  (\bibinfo{year}{2004}).

\bibitem[{\citenamefont{Ellis and MacCallum}(1969)}]{ellismac}
\bibinfo{author}{\bibfnamefont{G.}~\bibnamefont{Ellis}} \bibnamefont{and}
  \bibinfo{author}{\bibfnamefont{M.}~\bibnamefont{MacCallum}},
  \bibinfo{journal}{Comm. Math. Phys} \textbf{\bibinfo{volume}{12}},
  \bibinfo{pages}{108} (\bibinfo{year}{1969}).

\bibitem[{\citenamefont{Kohli and Haslam}(2013)}]{isk1}
\bibinfo{author}{\bibfnamefont{I.~S.} \bibnamefont{Kohli}} \bibnamefont{and}
  \bibinfo{author}{\bibfnamefont{M.~C.} \bibnamefont{Haslam}},
  \bibinfo{journal}{Phys. Rev. D} \textbf{\bibinfo{volume}{87}},
  \bibinfo{pages}{063006} (\bibinfo{year}{2013}),
  \urlprefix\url{http://link.aps.org/doi/10.1103/PhysRevD.87.063006}.

\bibitem[{\citenamefont{Ellis}(1973)}]{elliscargese}
\bibinfo{author}{\bibfnamefont{G.~F.} \bibnamefont{Ellis}},
  \emph{\bibinfo{title}{Cargese Lectures in Physics}}, vol.
  \bibinfo{volume}{Six} (\bibinfo{publisher}{Gordon and Breach},
  \bibinfo{year}{1973}), \bibinfo{edition}{1st} ed.

\bibitem[{\citenamefont{Hewitt et~al.}(2001)\citenamefont{Hewitt, Bridson, and
  Wainwright}}]{hewittbridsonwainwright}
\bibinfo{author}{\bibfnamefont{C.}~\bibnamefont{Hewitt}},
  \bibinfo{author}{\bibfnamefont{R.}~\bibnamefont{Bridson}}, \bibnamefont{and}
  \bibinfo{author}{\bibfnamefont{J.}~\bibnamefont{Wainwright}},
  \bibinfo{journal}{General Relativity and Gravitation}
  \textbf{\bibinfo{volume}{33}}, \bibinfo{pages}{65} (\bibinfo{year}{2001}).

\bibitem[{\citenamefont{Hervik et~al.}(2010)\citenamefont{Hervik, Lim, Sandin,
  and Uggla}}]{herviklim}
\bibinfo{author}{\bibfnamefont{S.}~\bibnamefont{Hervik}},
  \bibinfo{author}{\bibfnamefont{W.~C.} \bibnamefont{Lim}},
  \bibinfo{author}{\bibfnamefont{P.}~\bibnamefont{Sandin}}, \bibnamefont{and}
  \bibinfo{author}{\bibfnamefont{C.}~\bibnamefont{Uggla}},
  \bibinfo{journal}{Classical and Quantum Gravity}
  \textbf{\bibinfo{volume}{27}}, \bibinfo{pages}{185006}
  (\bibinfo{year}{2010}).

\bibitem[{\citenamefont{van Elst and Uggla}(1997)}]{vanelstuggla}
\bibinfo{author}{\bibfnamefont{H.}~\bibnamefont{van Elst}} \bibnamefont{and}
  \bibinfo{author}{\bibfnamefont{C.}~\bibnamefont{Uggla}},
  \bibinfo{journal}{Class. Quantum Grav.} \textbf{\bibinfo{volume}{14}},
  \bibinfo{pages}{2673} (\bibinfo{year}{1997}).

\bibitem[{\citenamefont{Tsagas and Maartens}(2000)}]{tsagasmaartens}
\bibinfo{author}{\bibfnamefont{C.~G.} \bibnamefont{Tsagas}} \bibnamefont{and}
  \bibinfo{author}{\bibfnamefont{R.}~\bibnamefont{Maartens}},
  \bibinfo{journal}{Classical and Quantum Gravity}
  \textbf{\bibinfo{volume}{17}}, \bibinfo{pages}{2215} (\bibinfo{year}{2000}).

\bibitem[{\citenamefont{Doroshkevich}(1965)}]{doroshkevich}
\bibinfo{author}{\bibfnamefont{A.}~\bibnamefont{Doroshkevich}},
  \bibinfo{journal}{Astrophysics} \textbf{\bibinfo{volume}{1}},
  \bibinfo{pages}{138} (\bibinfo{year}{1965}).

\bibitem[{\citenamefont{Thorne}(1967)}]{thorne2}
\bibinfo{author}{\bibfnamefont{K.~S.} \bibnamefont{Thorne}},
  \bibinfo{journal}{Astrophysical Journal} \textbf{\bibinfo{volume}{148}},
  \bibinfo{pages}{51} (\bibinfo{year}{1967}).

\bibitem[{\citenamefont{Hawking and Ellis}(2006)}]{ellis3}
\bibinfo{author}{\bibfnamefont{S.}~\bibnamefont{Hawking}} \bibnamefont{and}
  \bibinfo{author}{\bibfnamefont{G.}~\bibnamefont{Ellis}},
  \emph{\bibinfo{title}{The large scale structure of space-time}}
  (\bibinfo{publisher}{Cambridge University Press}, \bibinfo{year}{2006}),
  \bibinfo{edition}{twentieth printing} ed.

\bibitem[{\citenamefont{Coley and Wainwright}(1992)}]{coleywainwright}
\bibinfo{author}{\bibfnamefont{A.~A.} \bibnamefont{Coley}} \bibnamefont{and}
  \bibinfo{author}{\bibfnamefont{J.}~\bibnamefont{Wainwright}},
  \bibinfo{journal}{Class. Quantum Grav.} \textbf{\bibinfo{volume}{9}},
  \bibinfo{pages}{651} (\bibinfo{year}{1992}).

\bibitem[{\citenamefont{Hewitt and Wainwright}(1993)}]{hewwain}
\bibinfo{author}{\bibfnamefont{C.}~\bibnamefont{Hewitt}} \bibnamefont{and}
  \bibinfo{author}{\bibfnamefont{J.}~\bibnamefont{Wainwright}},
  \bibinfo{journal}{Classical and Quantum Gravity}
  \textbf{\bibinfo{volume}{10}}, \bibinfo{pages}{99} (\bibinfo{year}{1993}).

\end{thebibliography}

\end{document}